\documentclass[journal,onecolumn]{IEEEtran}

\usepackage{cite}

\ifCLASSINFOpdf
\usepackage[pdftex]{graphicx}
\else
\fi

\usepackage[cmex10]{amsmath}
\usepackage{amsthm}
\usepackage{amssymb}
\usepackage{bm}
\usepackage{latexsym}
\usepackage{bbm}
\usepackage{color}
\usepackage[subrefformat=parens]{subcaption}
\newtheorem{theorem}{Theorem}
\newtheorem{corollary}{Corollary}
\newtheorem{lemma}{Lemma}
\newtheorem{remark}{Remark}
\newtheorem{example}{Example}
\newtheorem{defn}{Definition}

\newcommand{\E}{\mathrm{E}}

\newcommand{\ds}{\displaystyle}
\newcommand{\st}{\mathrm{st}}
\newcommand{\PQD}{\mathrm{PQD}}
\newcommand{\Var}{\mathrm{Var}}
\newcommand{\Cov}{\mathrm{Cov}}
\newcommand{\sys}[1]{\langle #1\rangle}
\newcommand{\bms}[1]{\mbox{\scriptsize$\bm{#1}$}}

\allowdisplaybreaks

\begin{document}

\title{Effects of the Auto-Correlation of Delays on the Age of
Information: A Gaussian Process Framework}

\author{Atsushi Inoie and Yoshiaki Inoue
\thanks{This work was supported in part by JSPS KAKENHI Grant Number 24K14839.

A.\ Inoie is with 
Department of Information Network and Communication,
Kanagawa Institute of Technology, Atsugi-city 243-0292, Kanagawa, Japan. 
E-mail: inoie@nw.kanagawa-it.ac.jp.%

Y.\ Inoue is with 
Department of Information and Communications Technology, 
Graduate School of Engineering, The University of Osaka, Suita 565-0871, Japan.
E-mail: yoshiaki@comm.eng.osaka-u.ac.jp.}}%

\maketitle
\IEEEpeerreviewmaketitle

\begin{abstract}
The age of information (AoI) has been studied actively in recent years
as a performance measure for systems that require real-time
performance, such as remote monitoring systems via communication
networks. 
The theoretical analysis of the AoI is usually formulated based on
explicit system modeling, such as a single-server queueing model. 
However, in general, the behavior of large-scale systems such as
communication networks is complex, and it is usually difficult to
express the delay using simple queueing models. 
In this paper, we consider a framework in which the sequence of delays
is composed from a non-negative continuous-time stochastic process, called
a virtual delay process, as a new modeling approach for the theoretical
analysis of the AoI. 
Under such a framework, we derive an expression for the transient
probability distribution of the AoI and further apply the theory of
stochastic orders to prove that the high dependence of the sequence of
delays leads to the degradation of AoI performance. 
We further consider a special case in which the sequence of delays is
generated from a stationary Gaussian process, and we discuss the
sensitivity of the AoI to second-order statistics of the delay process
through numerical experiments.
\end{abstract}

\section{Introduction}\label{sec:intro}
The rapid advancement of information and communication technologies in
recent years has enabled new applications that rely on real-time data
processing. In particular, the widespread adoption of small and
cost-effective Internet of things (IoT) devices, facilitated by
miniaturization and low-power wireless technologies, has accelerated
the collection of massive amounts of data. By integrating these
large-scale data streams with artificial intelligence (AI) inference
and edge computing, future systems are expected to enable diverse
decision-making and control mechanisms based on the latest data in
real time.

Reflecting these trends, modern information and communication systems
are no longer limited to serving as ``dumb pipes'' that merely
transmit data. Instead, they function as dynamic platforms for
real-time information sharing, especially in applications like IoT
monitoring systems and intelligent transportation systems (ITS). In
these domains, timely and accurate information delivery significantly
enhances operational efficiency and safety. For instance, IoT
monitoring systems must handle continuous real-time data from sensors
deployed in smart environments, ranging from homes and factories to
public infrastructure, to ensure reliable oversight and immediate
responses. Similarly, ITS can leverage real-time data from vehicles
and roadside units to optimize traffic flow, reduce congestion, and
improve overall safety.

In such applications, the \textit{freshness of information} is a
critical factor that directly influences monitoring accuracy and
control responsiveness, ultimately affecting overall system
performance. Traditional performance metrics, such as throughput and
delay, do not fully capture the freshness requirement that underpins
real-time operations. To address this gap, the age of information
(AoI) has emerged as an essential metric for quantifying information
freshness \cite{Kaul11, Kaul12-1, Kosta2017, Yates2021}. The AoI
measures how much time has elapsed since the latest received update
was generated at the source, making it a direct indicator of
timeliness.

\begin{figure}[t]
\centering
\includegraphics[scale=0.4]{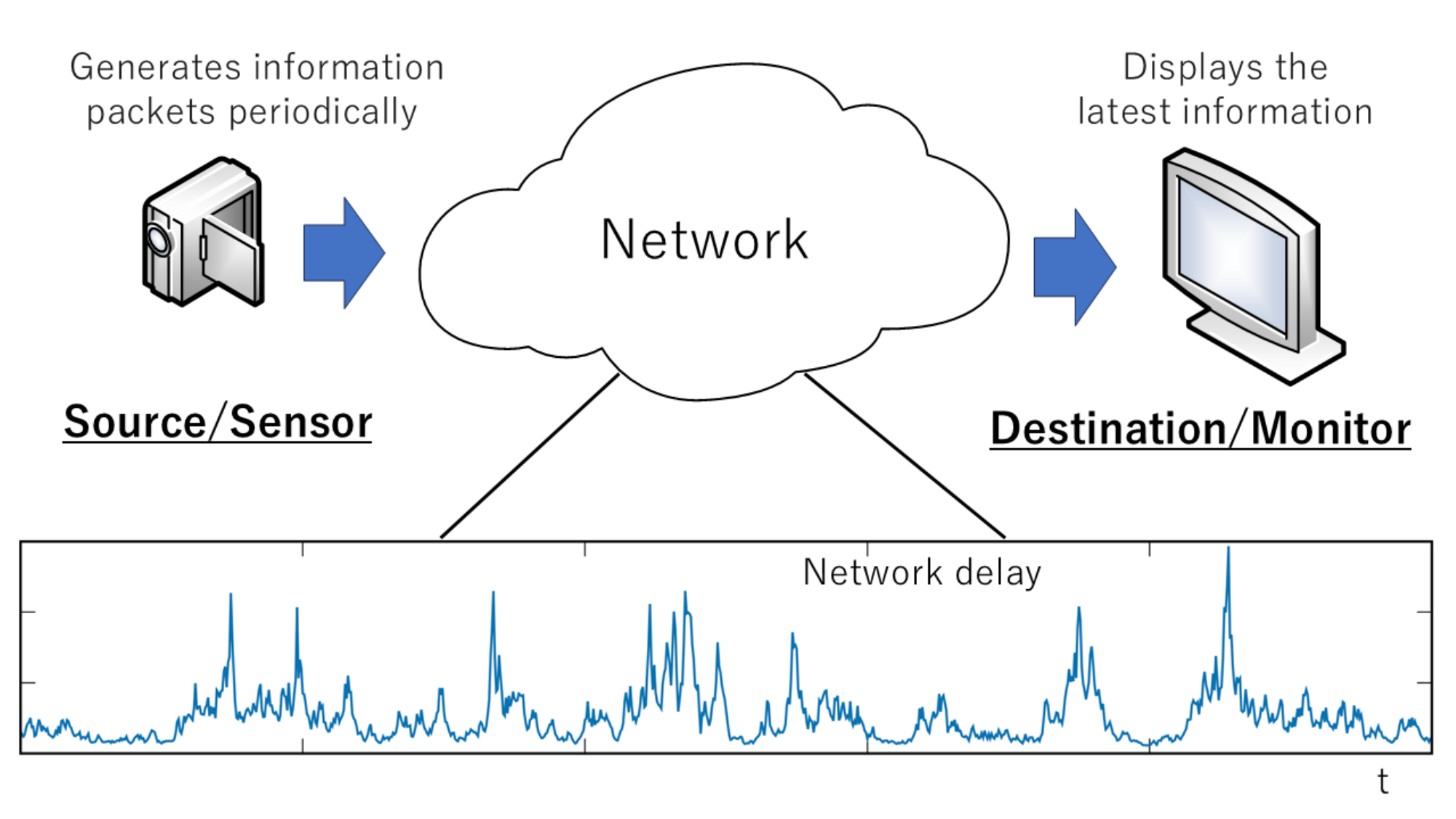} 
\caption{An example of a monitoring system.}
\label{fig:monitoring}
\end{figure}

Intuitively, the AoI depends on two key factors: the frequency of
information generation at the source and the end-to-end delay incurred
while transmitting information from the source to the monitor (Fig.~\ref{fig:monitoring}). When
the information generation rate is too low, the monitor may remain
idle for extended periods without new updates, leading to high AoI
(low freshness). Conversely, even if the generation rate is
sufficiently high, large delays cause updates to become stale before
reaching the monitor, thereby keeping the AoI at high values.

To analyze the AoI theoretically, we typically model the intergeneration
intervals $(G_n)_{n=0,1,\ldots}$ and the delay sequence
$(D_n)_{n=0,1,\ldots}$ as stochastic processes. The sequence $(G_n)_{n=0,1,\ldots}$
is often assumed to be independently and identically distributed
(i.i.d.) with nonnegative values, which is natural in sensing systems
(e.g., periodic sampling or Poisson sampling). However, modeling
$(D_n)_{n=0,1,\ldots}$ usually follows one of two approaches. The first considers
updated information queued in a single-server system until it is
received by the monitor
\cite{Kaul12-1,Costa16,Kam18,Bedewy19-1,Bedewy19-2,Inoue19,Yates19},
capturing congestion effects due to high update rates. The second
treats $(D_n)_{n=0,1,\ldots}$ as an i.i.d. sequence independent of
$(G_n)_{n=0,1,\ldots}$ \cite{Kam16,Yates2017,Inoue2021,Inoue2023}, 
which is often employed when analyzing scheduling or resource allocation problems
\cite{He16,Sun17-2,Tripathi17,Hsu20,Jiang18,He18,Kadota18,Kadota19,Li19,Yin19}.

Despite their utility, both approaches have potential limitations.
First, accurately capturing delay in large-scale networks via simple
queuing models can be challenging, as real-world systems often exhibit
complex behavior. Second, assuming i.i.d. delays and ignoring
autocorrelation can lead to overoptimistic performance estimates. For
example, in the i.i.d. model, making the intergeneration intervals
arbitrarily small may unrealistically suggest an infinitely small AoI.

To address these issues, we propose a framework in which the delay
sequence $(D_n)_{n=0,1,\ldots}$ is derived from a continuous-time nonnegative
stochastic process, referred to as a \textit{virtual delay process}. 
This framework generalizes the second approach above by allowing for
temporal dependence in delay. 
More specifically, the main contributions of this paper are summarized as follows:

\begin{itemize}
\item \textbf{New Modeling Framework:} We introduce a novel modeling framework in which the sequence of delays is generated by a nonnegative continuous-time stochastic process, referred to as a virtual delay process. Unlike existing approaches that rely on either single-server queueing or i.i.d.\ delay assumptions, this framework accommodates arbitrary dependence structures in the delay sequence.

\item \textbf{Analysis of the AoI:} Under our proposed model, we
derive the transient probability distribution of the AoI. 
Furthermore, by leveraging stochastic orders, we prove that stronger temporal dependence in the delay sequence degrades AoI performance. This result underscores the limitations of i.i.d.\ delay assumptions, which may be overly optimistic for large-scale or heavily correlated environments.

\item \textbf{Stationary Gaussian Case:} 
As a concrete example, we specialize our framework to the case in
which the sequence of delays is generated from a stationary Gaussian process. 
Through theoretical analysis and numerical evaluation, we explore how
the autocorrelation structure of the Gaussian process influence AoI.
\end{itemize}

This rest of the paper is organized as follows. Section~\ref{sec:relatedwork} provides a review of related work. Section~\ref{sec:model} presents our continuous-time delay modeling framework. In Section~\ref{sec:transient}, we derive the transient distribution of the AoI and analyze its properties. Section~\ref{sec:gaussian} discusses the stationary Gaussian case, and Section~\ref{sec:num} shows numerical experiments. Finally, we conclude this paper in Section~\ref{sec:conc}.

\section{Related Work}
\label{sec:relatedwork}

The study of the AoI has evolved into two major
research directions. The first direction, \emph{AoI analysis}, focuses
on deriving explicit expressions for the AoI (e.g., the mean and distribution) for
various system models, often leveraging queueing theory. The second,
\emph{AoI minimization}, is concerned with optimizing the AoI through
dynamic scheduling, sampling, or resource allocation mechanisms. While
conceptually distinct, these directions are tightly coupled in many
works, including the present one, which contributes to both the
theoretical analysis of the AoI and the development of
performance-improving insights.

\smallskip

\noindent
\textbf{AoI Analysis:}
Initial works in the AoI analysis considered single-source systems with
memoryless interarrival and/or service times. Kaul et al.\ derived
expressions for the average AoI in M/M/1, M/D/1, and D/M/1 queues under
first-come first-served (FCFS)~\cite{Kaul12-1}, and extended their analysis to last-come first-served (LCFS) 
policies~\cite{Kaul12-2}. Costa et al.~\cite{Costa16} demonstrated the
benefits of packet management under LCFS, where outdated packets are
discarded to improve freshness.

The role of service discipline in the performance of the AoI has since
been studied in depth. LCFS policies, especially those with
preemption, are known to achieve superior freshness in high-traffic
regimes. Bedewy et al.\cite{Bedewy19-2} examined multihop networks and
demonstrated that last-generated first-served (LGFS) scheduling strategies can achieve strong
AoI performance under certain service time distributions. 
Kam et al.~\cite{Kam18} analyzed systems with packet deadlines,
identifying scheduling policies that balance timeliness and complexity.
Inoue et al.~\cite{Inoue19} developed a general framework for
computing the stationary distribution of the AoI and analyzed several
single-server queues, including both FCFS and LCFS disciplines. Champati et
al.\ \cite{Campati19} analyzed single-server queues with generally
distributed inter-generation and service times, and derived bounds for
the AoI distribution under both FCFS and LCFS service discipline.

As attention shifted to multi-source systems, analytical challenges
grew due to interactions between competing update streams. Yates and
Kaul studied a two-stream M/M/1 queue~\cite{Yates19}, and their
results were extended to general service times in M/GI/1 models by
Moltafet et al.~\cite{Moltafet20} and Inoue and Takine~\cite{Inoue2024}.
Also, the GI/GI/1 queue with coexisting M/GI stream has been analyzed
by Inoue and Mandjes~\cite{Inoue2025}.
Chen et al.~\cite{Chen2022} addressed the
multi-stream M/G/1/1 case, highlighting the importance of buffer
constraints. Jiang and Miyoshi~\cite{Jiang2021} developed transform
techniques for the joint analysis of the AoI of multiple sources under
pushout discipline, while Liu et al.~\cite{Liu2021} offered robust
worst-case guarantees on the peak AoI using uncertainty sets.

Some works also extended the analysis to multi-server systems. Kam et
al.~\cite{Kam16} examined the AoI in M/M/2 and M/M/$\infty$ queues,
exploring how service parallelism improves information freshness.
Yates~\cite{Yates2020} developed a general stochastic hybrid system
(SHS) approach for analyzing the moments and distribution of the AoI
in network systems formulated as finite-state Markov chains.

Despite this progress, much of the existing literature relies on
simplifying assumptions such as exponential service times or Poisson
arrivals. Queueing
models, particularly single-server systems with FCFS or LCFS
disciplines, have served as the predominant analytical framework in
AoI research. Within this framework, packet delays are determined
implicitly by the service dynamics and queue discipline.

Alternatively, several studies adopt models in which the delay
sequence is represented as an i.i.d.\ process, independent of the
generation intervals~\cite{Kam16,Yates2017,Inoue2021,Inoue2023}. This
abstraction is often employed to facilitate tractable analysis and
optimization, especially in the design of sampling and scheduling
strategies under dynamic conditions. However, such models fail to
capture temporal dependencies in delays, which are common in systems
experiencing congestion or bursty traffic.

Even in classical queueing systems, delay processes can exhibit
autocorrelation as a consequence of queue buildup. But in such models,
the correlation is a resulting property rather than a controllable
parameter of models. As a result, existing models offer limited flexibility in
systematically exploring how delay dependence impacts the performance of the AoI.

\smallskip

\noindent
\textbf{AoI Minimization:}
Complementary to analysis-focused works, a significant line of
research seeks to minimize the AoI through system-level optimization.
He et al.~\cite{He18} analyzed the complexity of an AoI-oriented link
scheduling problem and introduced a mixed integer programming
formulation. Kadota et al.~\cite{Kadota18,Kadota19} proposed
scheduling algorithms for wireless networks, employing max-weight and
Whittle-index-based heuristics to balance throughput and freshness. 
Hsu et al.~\cite{Hsu20} considered the scheduling problem with random
arrivals and developed Whittle-index-based policies through
a Markov decision process (MDP) formulation.
Jiang et al.~\cite{Jiang18} investigated status updating over wireless
multiple-access channels with minimal coordination, while Li et
al.~\cite{Li19} developed general optimization frameworks for
heterogeneous edge networks. Yin et al.~\cite{Yin19} introduced a
concept of effective AoI, which focuses on update freshness at
request times, rather than continuously.

Joint sampling and scheduling has also been addressed. Sun et
al.~\cite{Sun2018} examined AoI-optimal update strategies across
multiple flows. Bedewy et al.~\cite{Bedewy2021} developed dynamic
programming-based approaches to simultaneously optimize sampling and
transmission timing. In energy-constrained networks, particularly
those powered by energy-harvesting sensors, a large body of
work~\cite{Bacinoglu15,Yates15,Arafa17,Bacinoglu17,Wu18,Arafa18,
Bacinoglu18,Farazi18, Hirosawa2020}
developed strategies that jointly balance energy usage and timeliness.

Recent efforts have expanded the scope of AoI optimization. Li et
al.~\cite{Li2022} addressed scheduling under AoI threshold
constraints, while Maatouk et al.~\cite{Maatouk2022} studied
lexicographic optimality in prioritized networks. Zakeri et
al.~\cite{Zakeri2024} incorporated deep reinforcement learning to
design policies that adapt to stochastic environments without full
model knowledge.

Nevertheless, these models often assume i.i.d.\ delays or
deterministic queueing behavior, limiting their realism in scenarios
with congestion that exhibits temporal correlation. This gap
motivates the need for models that explicitly handle and parameterize
delay dependence.

\smallskip

\noindent
\textbf{Contribution of This Paper:}
This paper addresses this gap by proposing a general framework in which
the delay sequence is modeled as a non-negative continuous-time
stochastic process. Unlike traditional models that assume independent
or memoryless delays, our framework captures temporal correlation and
allows for direct control over its strength. 
Also, unlike queueing-based approaches where correlation arises implicitly
from system dynamics, our model treats delay correlation as a tunable
property. This enables a systematic exploration of delay dependence
and provides insights beyond conventional single-server queues or i.i.d.\
delay settings. 

We derive exact characterizations of the time-dependent distribution
of the AoI, and analyze how increased delay dependence degrades the
AoI performance by means of the stochastic 
ordering techniques~\cite{shaked2006}.
We also show how the derived results apply to Gaussian
processes, revealing the sensitivity of the AoI to second-order statistics
of the delay process.

\section{System model}\label{sec:model}
Consider a system consisting of one sender (source/sensor) node and one receiver (destination/monitor) node. 
The sender node has a time-varying source of information, and the
receiver node has a monitor to display the information from the sender node. 
The sender node generates packets with a constant time interval $\tau$ ($\tau > 0$), and transmits the packets to the receiver node. 
The packets generated by the sender node are used for updating the information displayed at the receiver node. 

Let $\alpha_n := n\tau$ ($n = 0,1,\ldots$) denote the generation time of $n$-th packet, and
let $\beta_n = \alpha_n + D_n$ denote the arrival time of $n$-th
packet to the receiver node, where $D_n$ denotes 
a non-negative random variable representing the system delay
experienced when transmitting a packet from the sender node to the
receiver node. 
By definition of $\alpha_n$, we have $\alpha_0=0$. 
Recalling that $G_n$ denotes the interval of generation times between
$(n-1)$-th and $n$-th packets, we have $G_n=\alpha_n-\alpha_{n-1}=\tau$. 

In this paper, we consider the following framework: the sequence
$(D_n)_{n=0,1,\ldots}$ of system delays is characterized as
\begin{equation}
D_n = X_{\alpha_n} = X_{n\tau},
\quad
n = 0,1,\ldots,
\label{eq:D_n-by-X}
\end{equation}
where $(X_t)_{t \in [0,\infty)}$ denotes a non-negative continuous-time
stochastic process. In other words, the delay is given by the instantaneous value
of the process $(X_t)_{t \in [0,\infty)}$ at each generation time
$\alpha_n = n\tau$. 
We refer to $(X_t)_{t \in [0,\infty)}$ as {\it a virtual delay process}. 

\begin{remark}
In this framework, the virtual delay time process $(X_t)_{t \in
[0,\infty)}$ can be affected by generation events of packets. 
For example, in an FCFS queueing model, $(X_t)_{t \in [0,\infty)}$ 
is given by a stochastic process with upward jumps, where each jump size
corresponds to the service time of the packet, and it changes
its value at generation instants of packets.
\end{remark}
Let $\mathcal{I}_t := \{n;\, n=0,1,\ldots,\, \beta_n \leq t\}$ denote
the set of indices of packets arrived at receiver node until time $t$
($t \geq 0$) and let $J_t := \sup \mathcal{I}_t$ denote the  index of
the latest packet in $\mathcal{I}_t$. For convenience, we 
let $\sup\emptyset:=-\infty$ for the empty set $\emptyset$, 
so that we have $J_t = -\infty$ when $\mathcal{I}_t = \emptyset$. 

The AoI $A_t$ at time $t$ is then given by
\begin{equation*}\label{def:aoi}
A_t = t - \alpha_{J_t},
\quad
t \geq 0, 
\end{equation*}
where we define $\alpha_{-\infty} := -\infty$, so that the AoI is
defined as $A_t = \infty$ before the first arrival of a packet at the receiver node. 
The top panel of Fig.\ \ref{fig:exam_sample_path} shows sample paths
of $(X_t)_{t\ge 0}$ and $(D_n)_{n=0,1,\ldots}$,
while the bottom panel of Fig.\ \ref{fig:exam_sample_path} shows
a sample path of $(A_t)_{t\ge 0}$. 
\begin{figure}[t]
\centering
\includegraphics[scale=1.1]{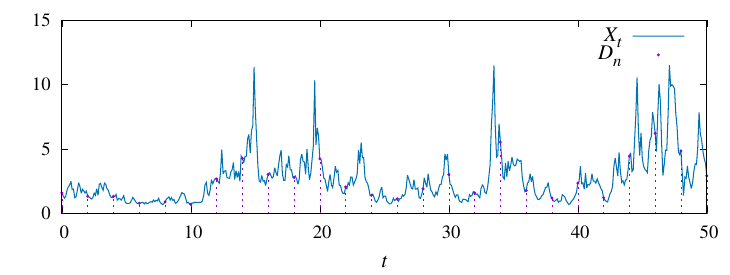} 
\includegraphics[scale=1.1]{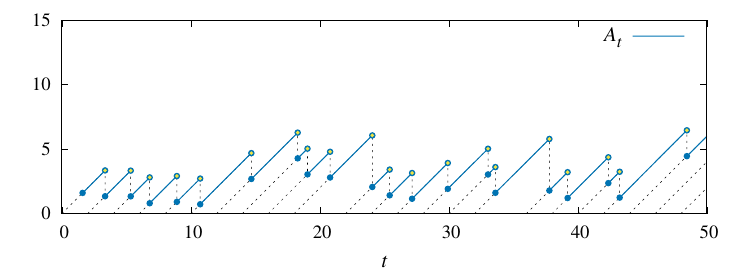} 
\caption{Top: An example of sample paths of $(X_t)_{t\ge 0}$ and
$(D_n)_{n=0,1,\ldots}$. Bottom: An example of a sample path of $(A_t)_{t\ge 0}$.
}\label{fig:exam_sample_path}
\end{figure}

\section{Transient analysis}\label{sec:transient}
Let $k_t$ denote the index of the latest packet generated before time $t$, 
and let $\phi_t$ denote the time elapsed
since the generation of the $k_t$-th packet:
\begin{equation}
k_t = \left\lfloor \frac{t}{\tau}\right\rfloor,
\;\;
\phi_t = t - k_t \tau,
\quad
t \geq 0. 
\label{eq:k_t-phi_t-def}
\end{equation}
By definition, we have $t = k_t \tau + \phi_t$ and $0 \leq \phi_t < \tau$.
Noting that $G_{J_t}=J_t\tau$, we have
\begin{align}
\Pr(A_t >x)
=
\Pr\left(\frac{t-x}{\tau} > J_t\right)
=
\Pr(J_t < \theta_t(x)),
\label{eq:AoI-CCDF-by-J}
\end{align}
where $\theta_t(x)$ ($t \geq 0$, $x \geq 0$) is given by
\begin{align}
\theta_t(x)=
\max
\left(0, \left\lceil \frac{t - x}{\tau}\right\rceil
\right)=
\max
\left(0,
k_t - \left\lfloor \frac{x-\phi_t}{\tau}\right\rfloor
\right). 
\label{eq:theta_t-def}
\end{align}
It follows from (\ref{eq:theta_t-def}) that $\theta_t(x) \leq k_t+1$
($x \geq 0$), where the equality holds if and only if $x < \phi_t$. 
We also note that $J_t\le k_t$. 
\begin{theorem}
\label{theorem:AoI-CCDF}
The complementary cumulative distribution function (CDF) of the AoI at time $t$ ($t \geq 0$) is given by
\begin{align}
\Pr(A_t > x)=
\left\{
\begin{array}{@{}l@{\;\;}l}
1, & x < \phi_t,
\\
\ds
\Pr\left(
\bigcap_{i=\theta_t(x)}^{k_t}\!\!
\{D_i > (k_t-i)\tau+\phi_t\} 
\right)\!\!, & x \geq \phi_t.\\
\end{array}
\right.
\label{eq:AoI-CCDF}
\end{align}
\end{theorem}
\begin{proof}
From (\ref{eq:AoI-CCDF-by-J}), $\Pr(A_t > x)$ is equivalent to the
probability that for all $i \in
\{\theta_t(x),\theta_t(x)+1,\ldots,k_t\}$,
the $i$-th packet does not arrive at the receiver node
until time $t$, i.e., $\beta_i > t$. 
For $x < \phi_t$, we have $\Pr(A_t > x)=1$ because $\theta_t(x) =
k_t+1$ holds in this case. 
For $x \geq \phi_t$, on the other hand, the arrival time of the $i$-th
packet is given by $\beta_i = i\tau + D_i$. 
Therefore, using $t = k_t\tau + \phi_t$, we see that $\beta_i > t$ is
equivalent to $D_i > (k_t-i)\tau + \phi_t$, which implies \eqref{eq:AoI-CCDF}.
\end{proof}
While the complementary CDF of the AoI shown in Theorem \ref{theorem:AoI-CCDF}
might look complicated, it enables us to derive some properties of the
transient AoI distribution as follows.

Let $D_n^{\min}$ denote the left endpoint of the support of the delay distribution:
\[
D_n^{\min} := \inf\{x \geq 0;\, \Pr(D_n \leq x) > 0\}, \quad
n = 0,1,\ldots.
\]

\begin{corollary}
\label{cor:A_t-x-refine}
The complementary CDF of the AoI at time $t$ ($t \geq 0$) is a step
function with respect to $x$. More specifically, $A_t$ takes its value
on a finite set $\{n\tau+\phi_t;\, n = j_t^*, j_t^*+1,\ldots, k_t\}\cup
\{\infty\}$ with probability $1$, and we have

\begin{equation}
\Pr(A_t = \infty) + \sum_{n=j_t^*}^{k_t} \Pr(A_t = n\tau+\phi_t) = 1,
\label{eq:A_t-finite}
\end{equation}
where
\[
j_t^* := \min\{j \in \{0,1,\ldots,k_t\};\, D_{k_t-j}^{\min} < j\tau+\phi_t\}.
\]
Furthermore, we have 
\begin{equation}
\Pr(A_t = \infty) = \Pr\left( \bigcap_{i=0}^{k_t} \{D_i > (k_t-i)\tau+\phi_t\} \right),
\label{eq:P-A_t-infty}
\end{equation}
which represents the probability that no packets have arrived at the receiver node until time $t$.
\end{corollary}
\begin{proof}
From (\ref{eq:theta_t-def}) and (\ref{eq:AoI-CCDF}), it is obvious
that $\Pr(A_t>x)$ ($t \geq 0$, $x \geq 0$) is a step function with
respect to $x$ and it can have discontinuities at $x \in \{n\tau + \phi;\, n=0,1,\ldots,k_t\}$. 

We next show that this set of the possible AoI values can be further narrowed down.
Note that we have $\Pr(D_i > (k_t-i)\tau+\phi_t)=1$ for $i \in \{0,1,\ldots,k_t\}$ such that $D_i^{\min} > (k_t-i)\tau + \phi_t$. 
Let $i_t^*$ denote the indice defined as
\[
i_t^* = \max\{i \in \{0,1,\ldots,k_t\};\, D_i^{\min} \leq (k_t-i)\tau + \phi_t\}. 
\]
The expression (\ref{eq:AoI-CCDF}) is then rewritten as
\begin{align}
\Pr(A_t > x)=
\left\{
\begin{array}{@{}l@{\;\;}l}
1, & x < \phi_t,
\\
\ds
\Pr\left(
\bigcap_{i=\theta_t(x)}^{i_t^*}\!\!
\{D_i > (k_t-i)\tau+\phi_t\} 
\right)\!\!, & x \geq \phi_t.\\
\end{array}
\right.
\label{eq:AoI-CCDF-2}
\end{align}
Note here that $\theta_t(x) \leq i_t^*$ for $x \geq (k_t-i_t^*)\tau + \phi$. 
From (\ref{eq:AoI-CCDF-2}), we can verify that the discontinuity
points of $\Pr(A_t>x)$ are included in the set $\{n\tau + \phi;\,
n=k_t-i_t^*, k_t-i_t^*+1,\ldots, k_t\}$, which implies (\ref{eq:A_t-finite})
because $j_t^* = k_t-i_t^*$. 

Also, noting that $\theta_t(x)=0$ holds for $x>t$,
we obtain (\ref{eq:P-A_t-infty}) directly from (\ref{eq:AoI-CCDF}).
\end{proof}

\begin{corollary}
If the sequence of system delays $(D_n)_{n=0,1,\ldots}$ is
stationary, the complementary CDF of the AoI satisfies
\begin{equation}
\Pr(A_t > x)=\Pr(A_{t+\tau}>x),
\qquad
t \geq x,
\label{eq:A_t-periodic}
\end{equation}
i.e., for $t \geq x$, it is a periodic function of $t$ with period $\tau$. 
\end{corollary}
\begin{proof}
From (\ref{eq:k_t-phi_t-def}), it is satisfied that
\begin{align*}
k_{t+\tau} &= \left\lfloor \frac{t+\tau}{\tau}\right\rfloor=k_t +1,\\
\phi_{t+\tau} &= t + \tau - k_{t+\tau} \tau= t + \tau - (k_t +1)\tau=t-k_t\tau=\phi_t. 
\end{align*}
Therefore, from (\ref{theorem:AoI-CCDF}), we have
\begin{align}
\Pr(A_{t+\tau} > x)&=
\left\{
\begin{array}{@{}l@{\;\;}l}
1, & x < \phi_{t+\tau},
\\
\ds
\Pr\left(
\bigcap_{i=\theta_{t+\tau}(x)}^{k_{t+\tau}}\!\!\!\!\!\!
\{D_i > (k_{t+\tau}-i)\tau+\phi_{t+\tau}\} 
\right)\!\!, & x \geq \phi_{t+\tau},\\
\end{array}
\right.
\nonumber
\\
&=
\left\{
\begin{array}{@{}l@{\;\;}l}
1, & x < \phi_t,
\\
\ds
\Pr\left(
\bigcap_{i=\theta_{t+\tau}(x)}^{k_t+1}\!\!\!\!\!\!
\{D_i > (k_t-i+1)\tau+\phi_t\} 
\right)\!\!, & x \geq \phi_t.\\
\end{array}
\right.
\label{eq:A_t+tau}
\end{align}
When $t \geq x$, it is satisfied that
\begin{align*}
\theta_t(x) = k_{t+\tau} - \left\lfloor \frac{x-\phi_{t+\tau}}{\tau}\right\rfloor=k_t + 1 - \left\lfloor \frac{x-\phi_t}{\tau}\right\rfloor=\theta_t(x)+1. 
\end{align*}
Therefore, 
\begin{align}
\Pr\left(
\bigcap_{i=\theta_t(x)+1}^{k_{t}+1} \{D_i > (k_{t}+1-i)\tau+\phi_t\} \right)
=\Pr\left(
\bigcap_{i=\theta_t(x)}^{k_{t}} \{D_{i+1} > (k_{t}-i)\tau+\phi_t\} \right). 
\label{eq:A_t+tau-x-large}
\end{align}
(\ref{eq:A_t+tau-x-large}) implied that if $(D_n)_{n=0,1,\ldots}$ is stationary, then 
$(D_{\theta_t(x)},\ldots,D_{k_t})$ and $(D_{\theta_t(x)+1}, \ldots,D_{k_t+1})$ have the same joint distribution. 
Therefore, from (\ref{theorem:AoI-CCDF}), (\ref{eq:A_t+tau}) and (\ref{eq:A_t+tau-x-large}), we obtain (\ref{eq:A_t-periodic}). 
\end{proof}

Finally, we prove that stronger dependency in the sequence of delays
$(D_n)_{n=0,1,\ldots}$ leads to an increase in the AoI. 
To that end, we introduce the following definitions of stochastic orders \cite{shaked2006}:
\begin{defn}
Let $Y^{\sys{1}}$ and $Y^{\sys{2}}$ $(Y^{\sys{1}}, Y^{\sys{2}} \in
[0,\infty))$ denote non-negative random variables. 
$Y^{\sys{1}}$ is said to be smaller than $Y^{\sys{2}}$ in the usual
stochastic order (denoted by $Y^{\sys{1}} \leq_{\st} Y^{\sys{2}}$) if and only if
\[
\Pr(Y^{\sys{1}} > x) \leq \Pr(Y^{\sys{2}} > x), \quad \forall x \geq 0.
\]
\end{defn}

\begin{defn}\label{defn:pqd}
Let $\bm{Y}^{\sys{1}}$ and $\bm{Y}^{\sys{2}}$ $(\bm{Y}^{\sys{1}},
\bm{Y}^{\sys{2}} \in \mathbb{R}^{\ell})$ denote $\ell$-dimensional random
vectors. 
$\bm{Y}^{\sys{1}}$ is said to be smaller than $\bm{Y}^{\sys{2}}$ in the positive
quadrant dependent (PQD) order (denoted by $\bm{Y}^{\sys{1}}
\leq_{\PQD} \bm{Y}^{\sys{2}}$) if and only if
\begin{align*}
\Pr(\bm{Y}^{\sys{1}} \leq \bm{x}) &\leq \Pr(\bm{Y}^{\sys{2}} \leq \bm{x}),
\quad
\forall \bm{x} \in \mathbb{R}^{\ell}
,
\end{align*}
and 
\begin{align*}
\Pr(\bm{Y}^{\sys{1}} > \bm{x}) &\leq \Pr(\bm{Y}^{\sys{2}} > \bm{x}),
\quad
\forall \bm{x} \in \mathbb{R}^{\ell}.
\end{align*}

\end{defn}
The PQD order is used to compare the strength of dependence between multidimensional random variables.
By definition \ref{defn:pqd}, $\bm{Y}^{\sys{1}} \leq_{\PQD} \bm{Y}^{\sys{2}}$
implies that $\bm{Y}^{\sys{1}}$ and $\bm{Y}^{\sys{2}}$ have the same marginal distribution. 

The following result is immediate from Theorem \ref{theorem:AoI-CCDF}
and Definition \ref{defn:pqd}: 
\begin{theorem}
\label{theorem:PQD-AoI}
Consider two systems with different sequences of delays. 
In each of these systems, the sender node generates packets at time $\alpha_n = n\tau$. 
Let $D_i^{\sys{j}}$ $(i=1,\ldots,m,\ j=1,2)$ denote the delay of $i$-th packet in system $j$ and $A_t^{\sys{j}}$ $(i=1,\ldots,m,\ j=1,2)$ denote the AoI in system $j$. 
If we have
\begin{align}
(D_0^{\sys{1}}, D_1^{\sys{1}}, \ldots, D_m^{\sys{1}})
\!\leq_{\PQD}\!
(D_0^{\sys{2}}, D_1^{\sys{2}}, \ldots, D_m^{\sys{2}}),
\label{eq:D-PQD}
\end{align}
for any $m=0,1,\ldots$, then $A_t^{\sys{1}} \leq_{\st} A_t^{\sys{2}}$ ($t \geq 0$) holds.  
\end{theorem}
Theorem \ref{theorem:PQD-AoI} indicates that the AoI increases with stronger dependence in the sequence of delays. In particular, the system with i.i.d. delays (i.e., an infinite-server queueing system) achieves the lowest AoI among all systems with the same marginal delay distribution.

\section{Special case: model with the sequence of delays composed from a stationary Gaussian process}
\label{sec:gaussian}

As a concrete example of the framework presented above, we consider a
model in which the virtual delay process $(X_t)_{t \geq 0}$ is defined
as follows. Let $(Z_t)_{t \geq 0}$ denote a stationary Gaussian process
with mean zero and unit variance. By the properties of stationary
Gaussian processes, the distribution of $(Z_t)_{t \geq 0}$ is uniquely
determined by its autocovariance function, given by $\hat{\sigma}(t)
:= \Cov[Z_0, Z_t] = \E[Z_0 Z_t]$ for $t \geq 0$.

More specifically, let $\bm{\Sigma}_{t_0,t_1,\ldots,t_{n-1}}$ denote an
$n \times n$ matrix whose $(i,j)$-th element is given by
\[
[\bm{\Sigma}_{t_0,t_1,\ldots,t_{n-1}}]_{i,j}
=
\hat{\sigma}(t_i-t_j),
\quad
i,j \in \{0,1,\ldots,n-1\}.
\]
The random vector $\bm{Z} = (Z_{t_0}, Z_{t_1}, \ldots, Z_{t_{n-1}})$ ($n=0,1,
\ldots$) then follows an $n$-dimensional normal distribution with mean
vector $\bm{0}$ and covariance matrix $\bm{\Sigma}:=\bm{\Sigma}_{t_0,t_1,\ldots,t_{n-1}}$;
the joint density function of $\bm{Z}$ is given by
\[
f_{\bms{Z}}(\bm{x})
= 
\frac{1}{\sqrt{(2\pi)^n \det(\bm{\Sigma)}}}
\exp\left(-\frac{1}{2} \bm{x}\bm{\Sigma^{-1}}\bm{x}^{\top}\right). 
\]

In this section, we define a model in which the virtual delay process $(X_t)_{t \geq 0}$ is given by
\begin{equation}
X_t = g(Z_t),
\quad
t \geq 0,
\label{eq:X_t-by-Z}
\end{equation}
where $g : \mathbb{R} \to [0,\infty)$ denotes a non-negative, monotonically non-decreasing function.
We refer to $(X_t)_{t \geq 0}$ as {\it a virtual delay process
composed from the stationary Gaussian process $(Z_t)_{t \geq 0}$}. 
Let $\mu := \E[X_t]$ and $\sigma(t) := \Cov[X_0, X_t]$ denote the mean and autocovariance function of $X_t$. 
\begin{example}
If we set $g(x) = x_{\min} + \exp(\hat{\mu} + \hat{s}x)$ ($x_{\min} \geq 0$, $\hat{\mu} \in
\mathbb{R}$, $\hat{s} > 0$), then $X_t - x_{\min}$ follows a
log-normal distribution, i.e., $X_t$ follows a log-normal distribution
with the left-end point $x_{\min}$. 
\end{example}

\begin{example}
If we set $g(x) = \max(x_{\min}, \hat{\mu} + \hat{s}x)$ ($x_{\min} \geq 0$,
$\hat{\mu} \in \mathbb{R}$, $\hat{s} > 0$), then $X_t$ follows a censored normal distribution.
\end{example}
Note that $(X_t)_{t \geq 0}$ is also a stationary process due to the stationarity of $(Z_t)_{t \geq 0}$. 
Fig. \ref{fig:z_sim} shows sample paths of $(X_t)_{t \geq 0}$
characterized by $g$ introduced in Examples 1 and 2. 
\begin{figure}[t!]
\centering
\includegraphics[scale=0.5]{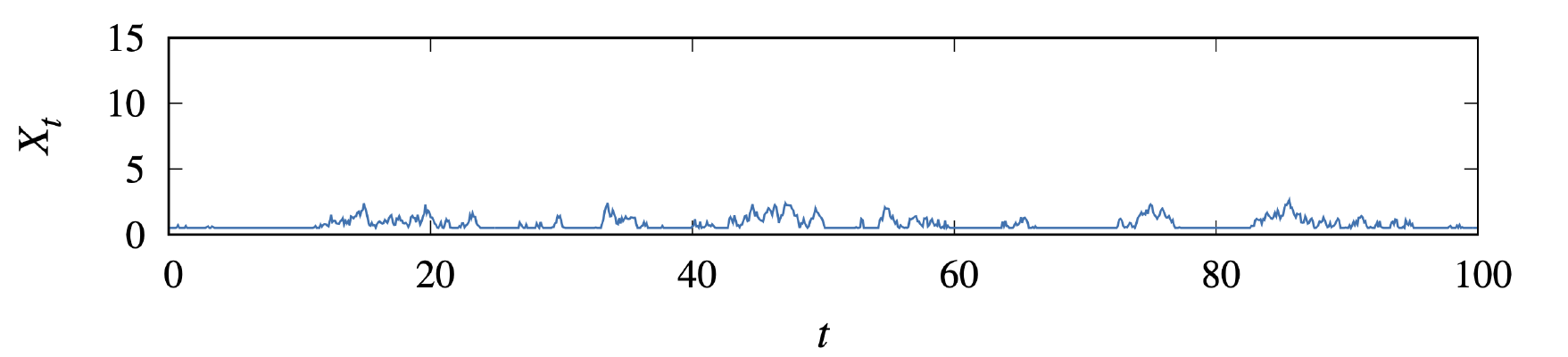} 
\includegraphics[scale=0.5]{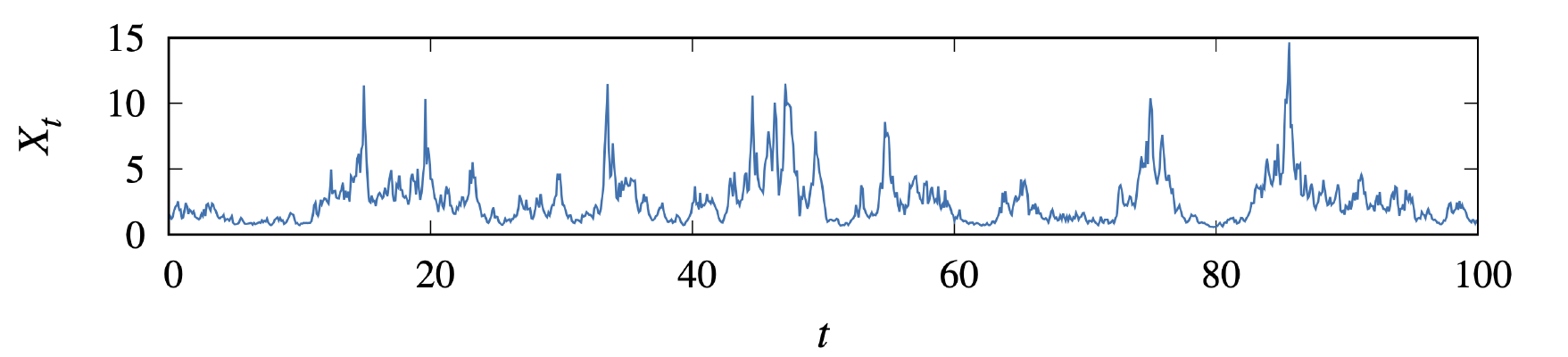} 
\caption{Sample paths of virtual delay processes $(X_t)_{t\ge 0}$
composed from a stationary Gaussian process (Top: Example 1, Bottom: Example 2)}\label{fig:z_sim}
\end{figure}

\medskip

In this model, Theorem \ref{theorem:PQD-AoI} simplifies as follows:
\begin{theorem}
\label{theorem:AoI-order-gauss}
Consider two systems with sequences of delays composed from stationary Gaussian processes with autocovariance functions $\hat{\sigma}^{\sys{1}}(t)$ and $\hat{\sigma}^{\sys{2}}(t)$, respectively. 
In each of these systems, the sender node generates packets at time $\alpha_n = n\tau$. 
Let $A_t^{\sys{j}}$ denote the AoI in system $j$, and let $g$ denote
the monotonically non-decreasing function used to compose sequences of
delays in both systems. We then have the following: 
\begin{itemize}
\item[(i)]
If the autocovariance functions of the stationary Gaussian processes
$(Z^{\sys{1}}_t)_{t \geq 0}$ and $(Z^{\sys{2}}_t)_{t \geq 0}$ satisfy
\[
\hat{\sigma}^{\sys{1}}(t) \leq \hat{\sigma}^{\sys{2}}(t),
\quad
\forall t \geq 0,
\]
then $A_t^{\sys{1}}\leq_{\st} A_t^{\sys{2}}$ for $t \geq 0$.
\item[(ii)] 
If the autocovariance functions of the virtual delay process
$(X^{\sys{1}}_t)_{t \geq 0}$ and $(X^{\sys{2}}_t)_{t \geq 0}$ satisfy
\[
\sigma^{\sys{1}}(t) < \sigma^{\sys{2}}(t),
\quad
\forall t > 0.
\]
then $A_t^{\sys{1}}\leq_{\st} A_t^{\sys{2}}$ for $t \geq 0$. 
\end{itemize}
\end{theorem}

\begin{proof}
We first prove (i).
From \cite[Example 9.A.8]{shaked2006}, $\hat{\sigma}^{\sys{1}}(t) \leq \hat{\sigma}^{\sys{2}}(t)$
implies
\[
(Z_0^{\sys{1}}, Z_{\tau}^{\sys{1}}, \ldots, Z_{m\tau}^{\sys{1}}) \leq_{\PQD}
(Z_0^{\sys{2}}, Z_{\tau}^{\sys{2}}, \ldots, Z_{m\tau}^{\sys{2}}),
\]
for $m=0,1,\ldots$. 
Because the PQD order is preserved under transformation by
the monotone function $g$ (\cite[Theorem 9.A.4]{shaked2006}), we can
verify from (\ref{eq:D_n-by-X}) and (\ref{eq:X_t-by-Z}) that
(\ref{eq:D-PQD}) holds for $m=0,1,\ldots$.
Therefore, we obtain (i) from Theorem \ref{theorem:PQD-AoI}. 

Next, we prove (ii). Similarly to (i), it can be readily shown that 
$\hat{\sigma}^{\sys{1}}(t) \geq \hat{\sigma}^{\sys{2}}(t)$ implies $(X_0^{\sys{1}}, X_t^{\sys{1}}) \geq_{\PQD} (X_0^{\sys{2}}, X_t^{\sys{2}})$. Therefore, we have
\[
\hat{\sigma}^{\sys{1}}(t) \geq \hat{\sigma}^{\sys{2}}(t)
\ \Rightarrow\
\sigma^{\sys{1}}(t) \geq \sigma^{\sys{2}}(t).
\]
Since the contraposition 
\[
\sigma^{\sys{1}}(t) < \sigma^{\sys{2}}(t)
\ \Rightarrow\
\hat{\sigma}^{\sys{1}}(t) < \hat{\sigma}^{\sys{2}}(t),
\]
is also valid, we obtain (ii) from (i). 
\end{proof}

Let $g^{-1}(y) := \inf\{x;\, g(x) > y\}$ ($y \geq 0$) denote the inverse function of $g$.
The following theorem shows that $\Pr(A_t > x)$ can be computed by multiple integrations of the joint density function of multi-dimensional normal distributions: 
\begin{theorem}
\label{thm:AoI-CCDF-Gauss}
The complementary CDF of the AoI at time $t$ ($t \geq 0$) is given by
\begin{align}
\Pr(A_t > x)
=
\left\{
\begin{array}{@{}l@{\;\;}l}
1, & x < \phi_t
\\[1ex]
\ds
\Pr\left(
\bigcap_{i=\theta_t(x)}^{k_t}\{Z_{i\tau}>a_{i,t}\}
\right),
& x \geq \phi_t,
\end{array}
\right.
\label{eq:AoI-CCDF-Gauss}
\end{align}
where $a_{i,t}$ is defined as 
\[
a_{i,t} = g^{-1}\bigl((k_t-i)\tau+\phi_t\bigr).
\]
\end{theorem}
\begin{proof}
Using (\ref{eq:X_t-by-Z}), we rewrite (\ref{eq:AoI-CCDF}) as
\begin{align*}
\Pr(A_t > x)
&=
\left\{
\begin{array}{@{}l@{\;\;}l}
1, & x < \phi_t
\\
\ds
\Pr\left(
\bigcap_{i=\theta_t(x)}^{k_t}\!\!
\{g(Z_{i\tau}) > (k_t-i)\tau+\phi_t\} 
\right)\!\!, & x \geq \phi_t\\
\end{array}
\right.
\\
&=
\left\{
\begin{array}{@{}l@{\;\;}l}
1, & x < \phi_t
\\
\ds
\Pr\left(
\bigcap_{i=\theta_t(x)}^{k_t}\!\!
\{Z_{i\tau} > g^{-1}\bigl((k_t-i)\tau+\phi_t\}\bigr)
\right)\!\!, & x \geq \phi_t,
\end{array}
\right.
\end{align*}
which proves Theorem \ref{thm:AoI-CCDF-Gauss}. 
\end{proof}

\begin{figure*}[t]
\centering
\includegraphics[scale=0.19]{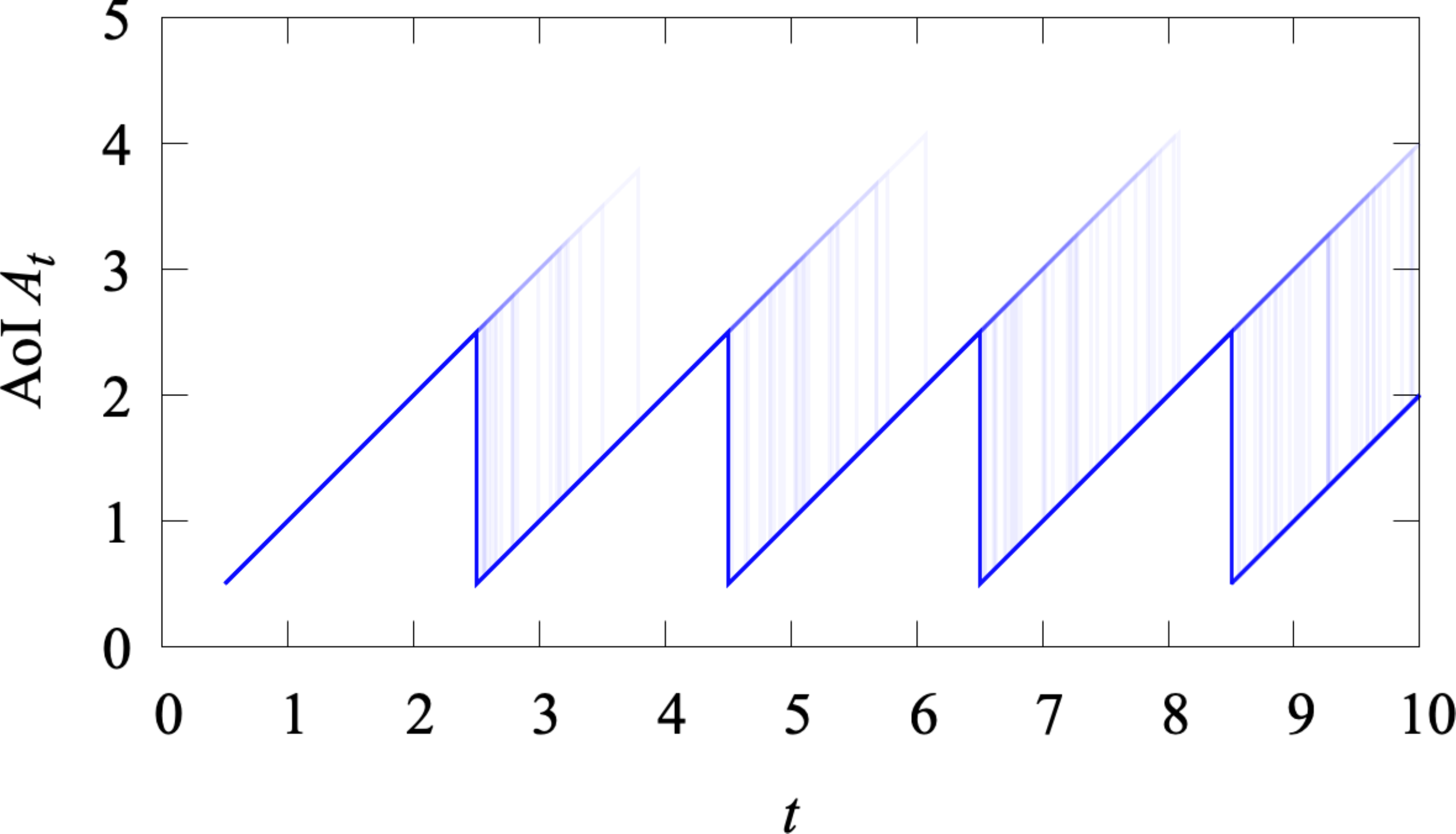} 
\includegraphics[scale=0.19]{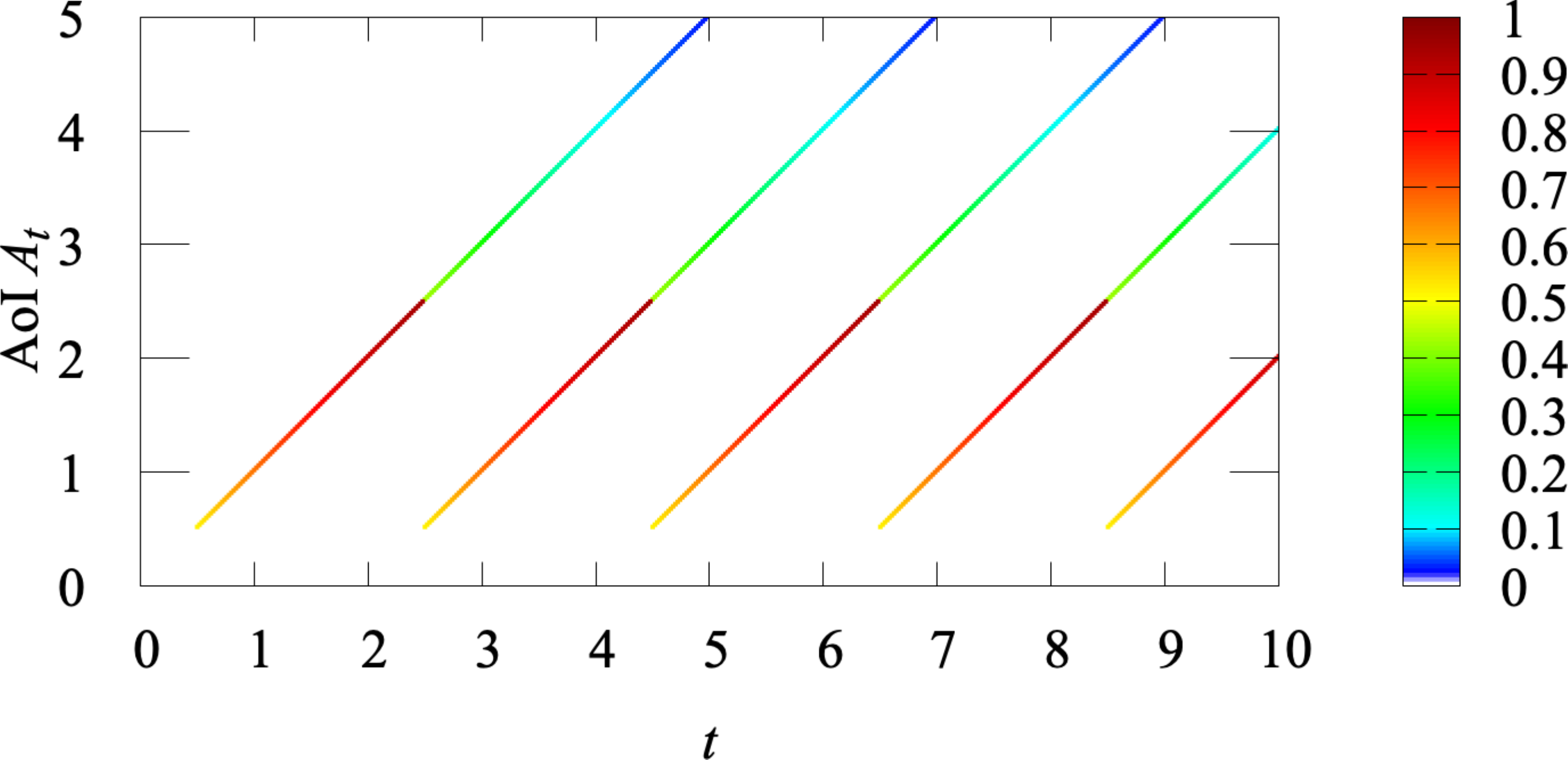} 
\caption{Simulation results of the AoI for Example 1 (left) and heat map
of probability mass function (right): $s=0.75$, $c=10$, $\tau=2.0$.}\label{fig:aoi_sim:ex1}
\includegraphics[scale=0.19]{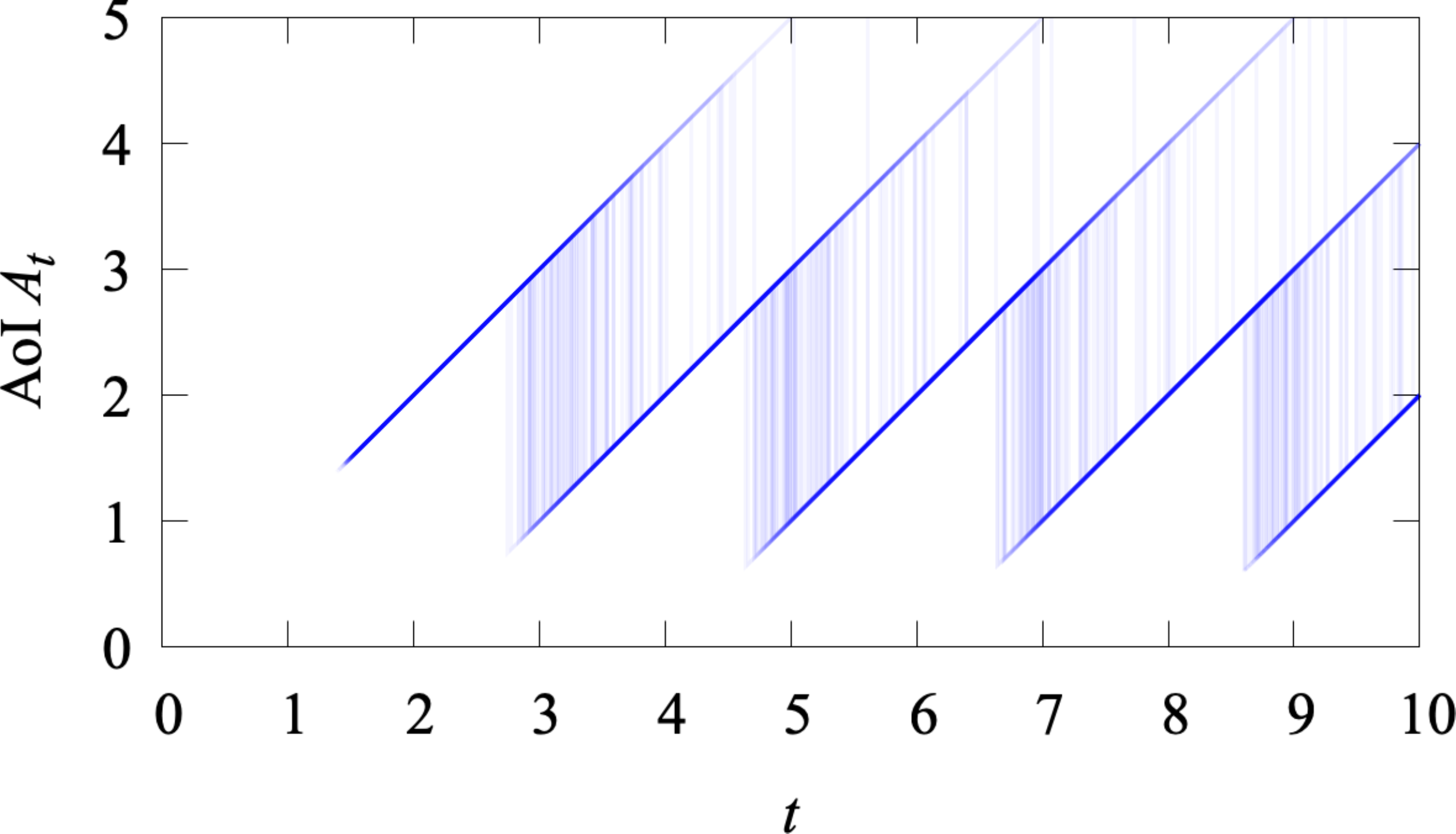} 
\includegraphics[scale=0.19]{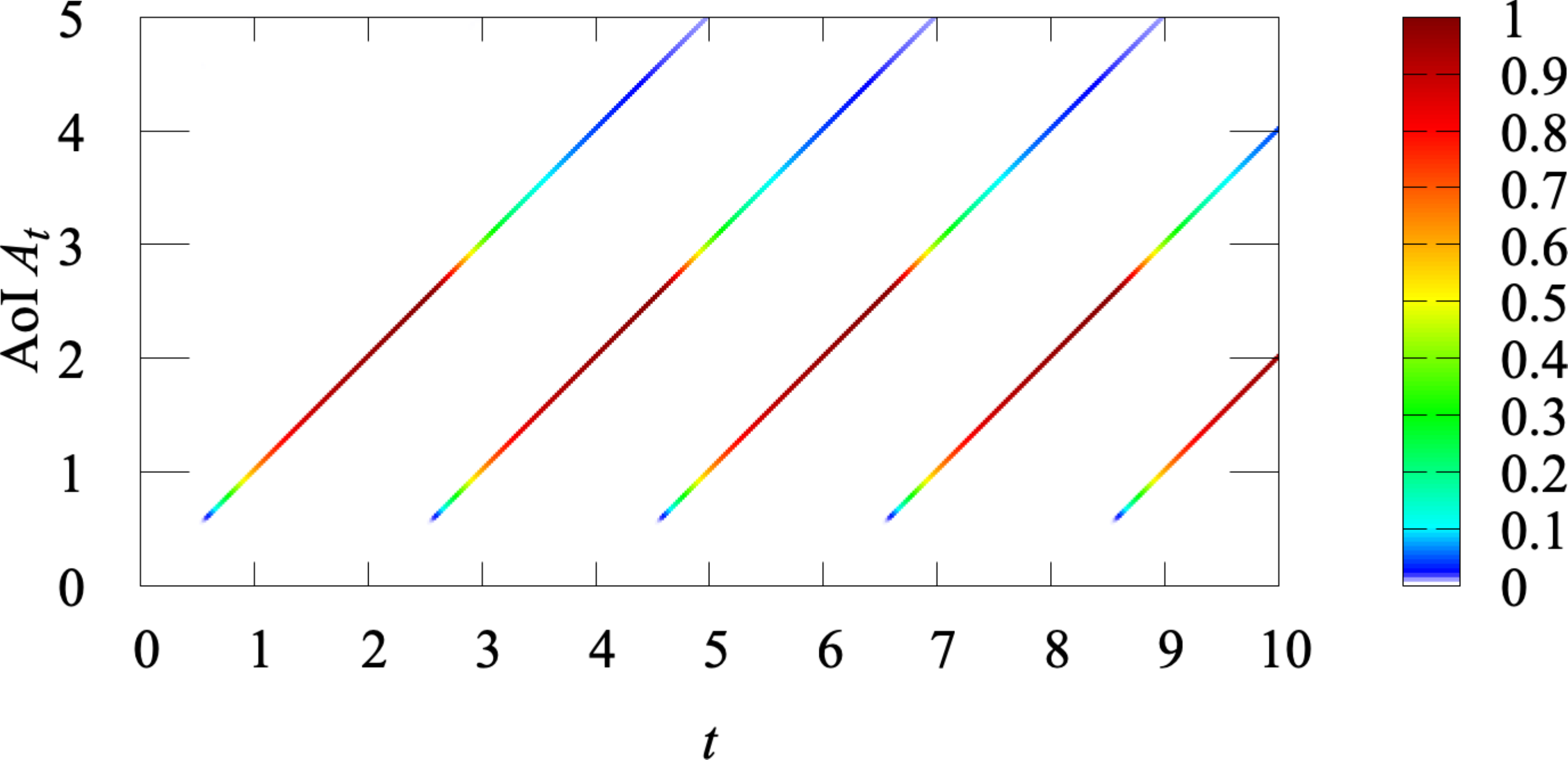} 
\caption{Simulation results of the AoI for Example 2 (left) and heat map
of probability mass function (right): $s=0.75$, $c=10$, $\tau=2.0$. }\label{fig:aoi_sim:ex2}
\end{figure*}

A well-known example of a stationary Gaussian process $(Z_t)_{t \geq 0}$ is 
the stationary Ornstein-Uhlenbeck (OU) process. 
The OU process is a Gaussian process that evolves over time according to the
following stochastic differential equation:
\[
{\textrm d}Z_t=-\kappa Z_t{\textrm d}t+ \sqrt{2\kappa}{\textrm d}W_t,
\]
where $\kappa$ denotes a positive parameter and $W_t$ denotes a
standard Wiener process.
Because the stationary OU process has the Markov property, it is also called the Gauss-Markov process. 
Notice that for any $\kappa > 0$, $Z_t$ follows a normal distribution
with mean $0$ and variance $1$, independent of $t$.

Moreover, conditional on $Z_0 = z$, the transient value $Z_t$ follows a normal
distribution with mean $ze^{-\kappa t}$ and variance $1-e^{-2\kappa t}$. 
In other words, letting $f_{N(\tilde{\mu},\tilde{s}^2)}(x)$ denote
the density function of the normal distribution with mean $\tilde{\mu}$ and variance $\tilde{s}^2$
\[
f_{N(\tilde{\mu},\tilde{s}^2)}(x) 
:= 
\frac{1}{\sqrt{2\pi }\tilde{s}}
\exp\left(-\frac{1}{2} \cdot \left(\frac{x-\tilde{\mu}}{\tilde{s}}\right)^2 \right),
\]
we have
\begin{equation}
\frac{{\rm d}}{{\rm d}x}
\Pr(Z_t \leq x \mid Z_0 = z)
=
f_{N(ze^{-\kappa t},1-e^{-2\kappa t})}(x). 
\label{eq:OU-transient-density}
\end{equation}
Therefore, the autocovariance function of the stationary OU process is given by 
\[
\hat{\sigma}(t)=E[Z_0 Z_t] = e^{-\kappa t},
\]
which is decreasing with respect to both $t$ and $\kappa$.

Utilizing the Markov property of the stationary OU process $(Z_t)_{t
\geq 0}$, the computation of the AoI distribution in
(\ref{eq:AoI-CCDF-Gauss}) can be improved in efficiency. See Appendix \ref{appendix:OU-compute} for details. 

\begin{figure*}[ht]
\centering
\begin{minipage}[b]{0.49\linewidth}
\includegraphics[scale=0.17]{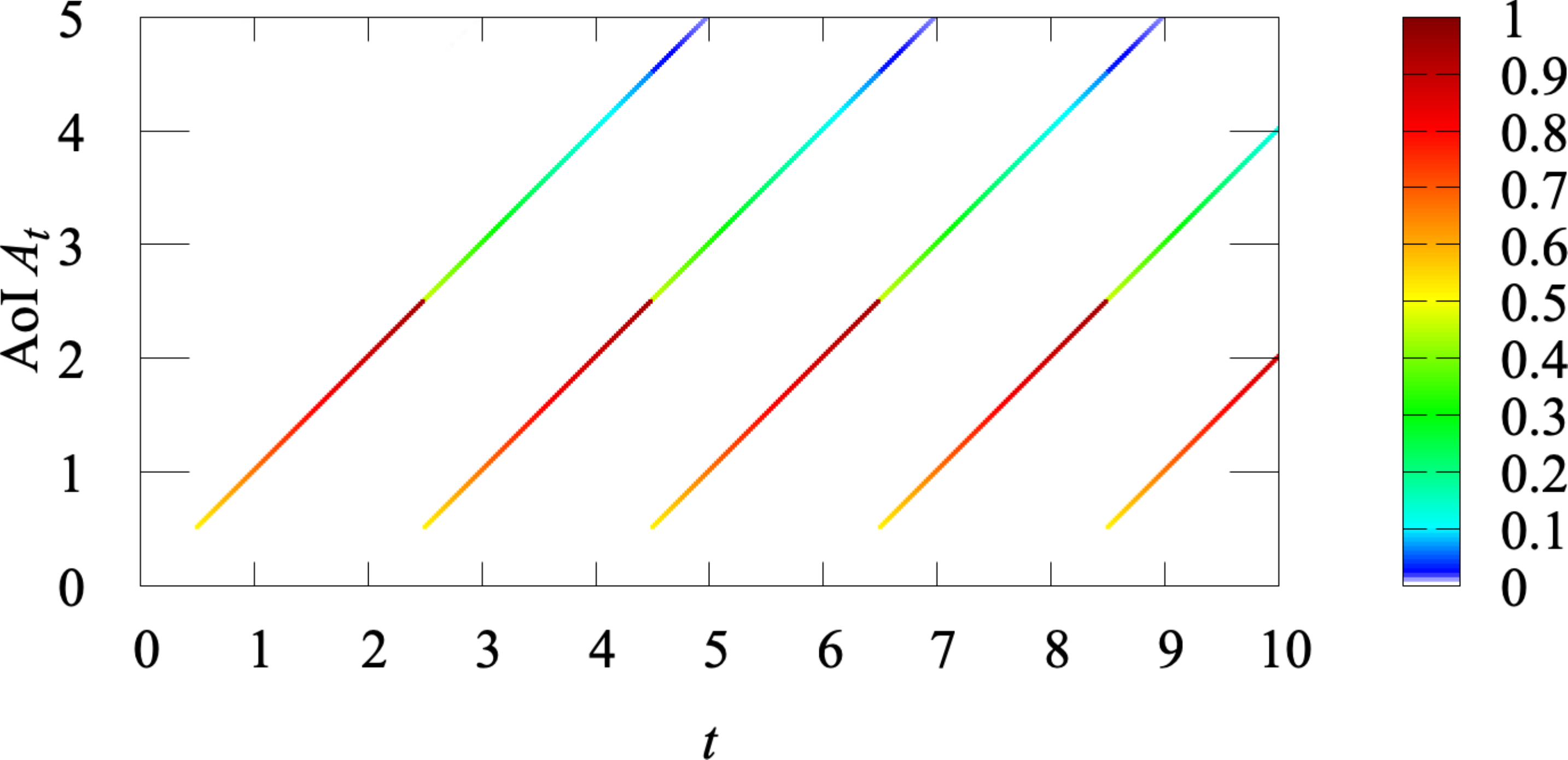} 
\subcaption{$c=0.1$, $\tau=2.0$}
\end{minipage}
\begin{minipage}[b]{0.49\linewidth}
\includegraphics[scale=0.17]{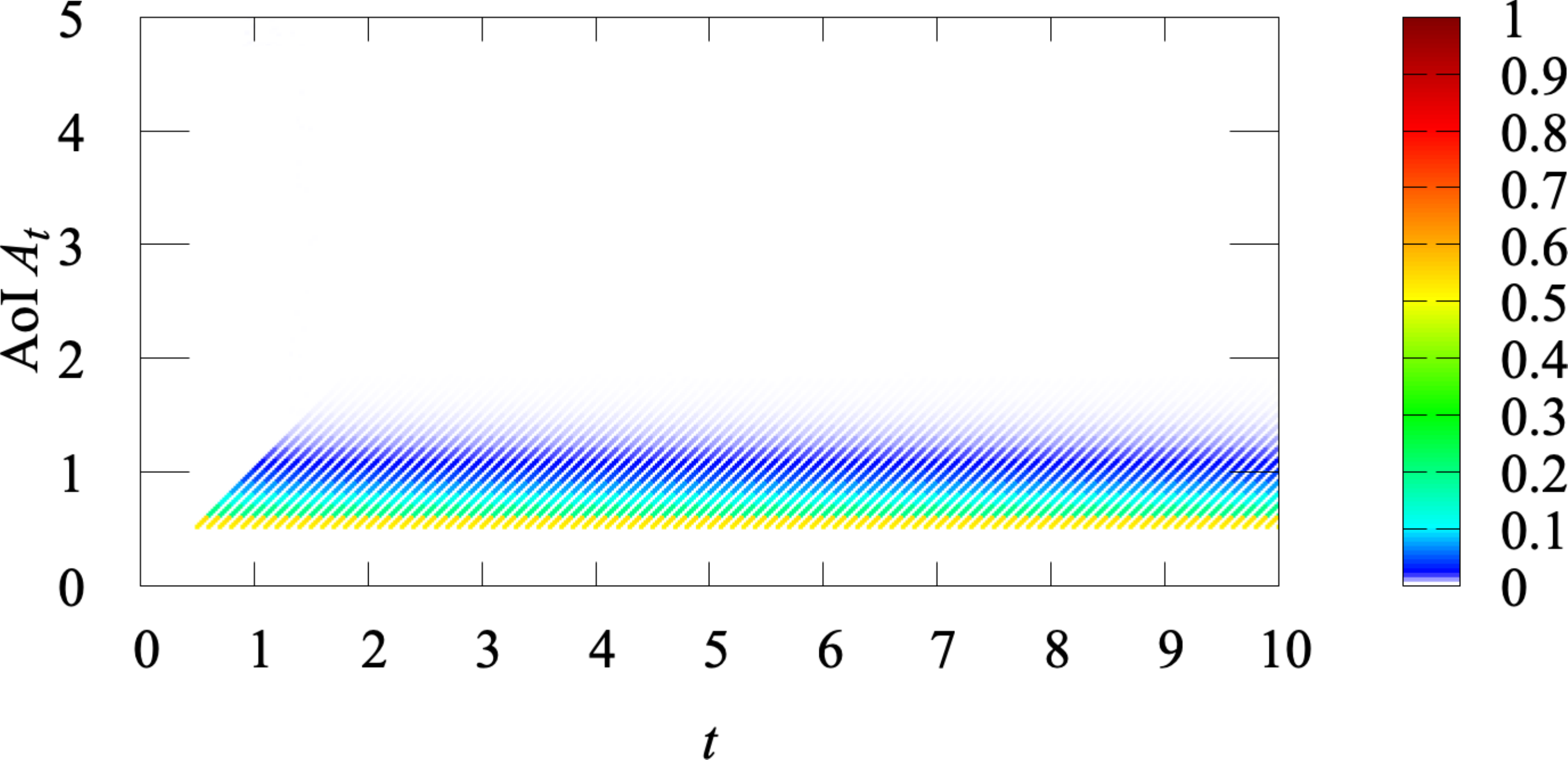} 
\subcaption{$c=0.1$, $\tau=0.1$}
\end{minipage}
\begin{minipage}[b]{0.49\linewidth}
\includegraphics[scale=0.17]{figs/r12_1hmpmf20a.pdf} 
\subcaption{$c=10.0$, $\tau=2.0$}
\end{minipage}
\begin{minipage}[b]{0.49\linewidth}
\includegraphics[scale=0.17]{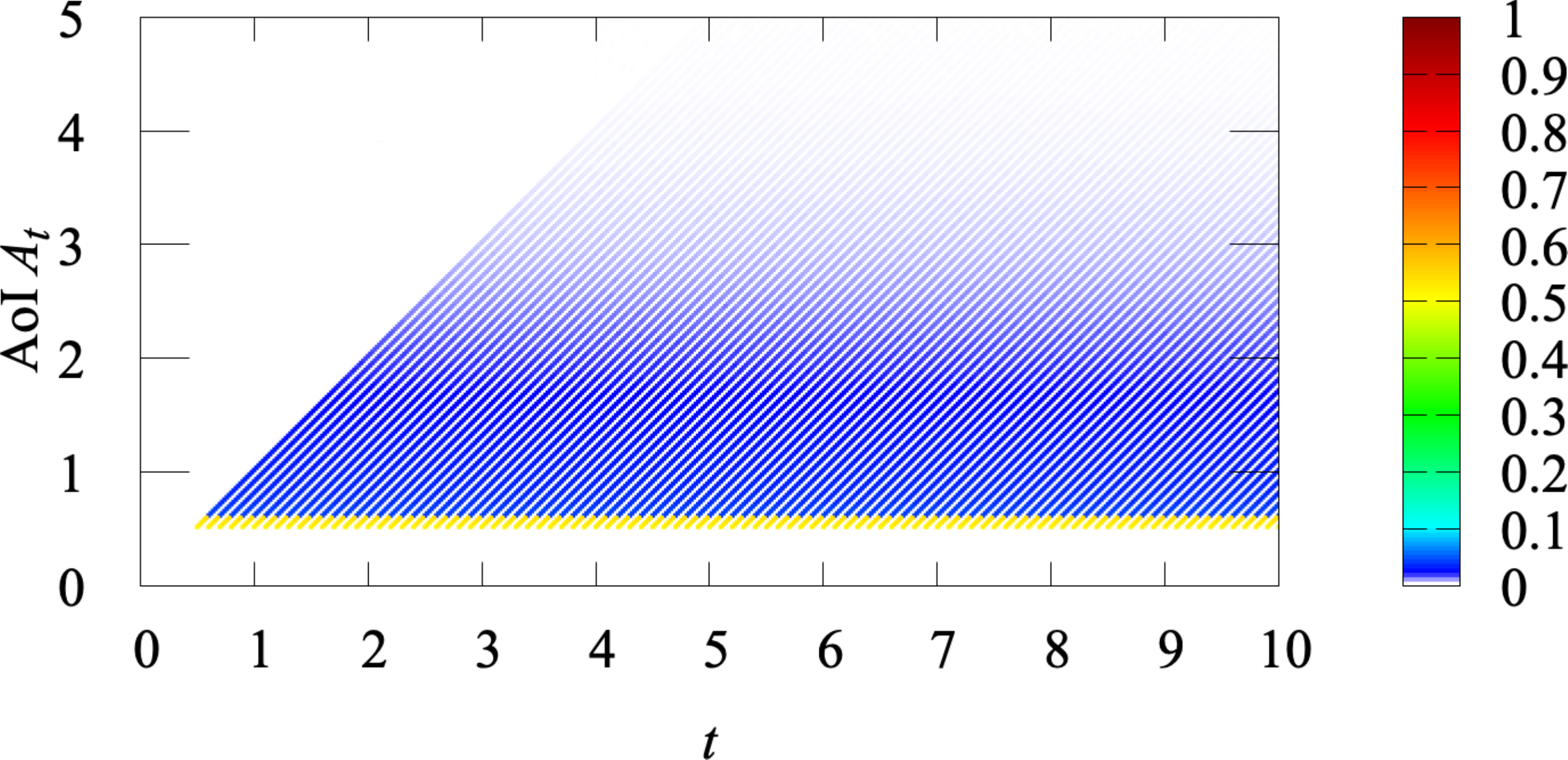} 
\subcaption{$c=10.0$, $\tau=0.1$}
\end{minipage}
\caption{Heatmap of the probability mass function of the AoI distribution in Example 1 ($s=0.75$).}\label{fig:tau_vs_heatmap3_1_A}
\end{figure*}
\begin{figure*}[ht!]
\centering
\begin{minipage}[b]{0.49\linewidth}
\includegraphics[scale=0.17]{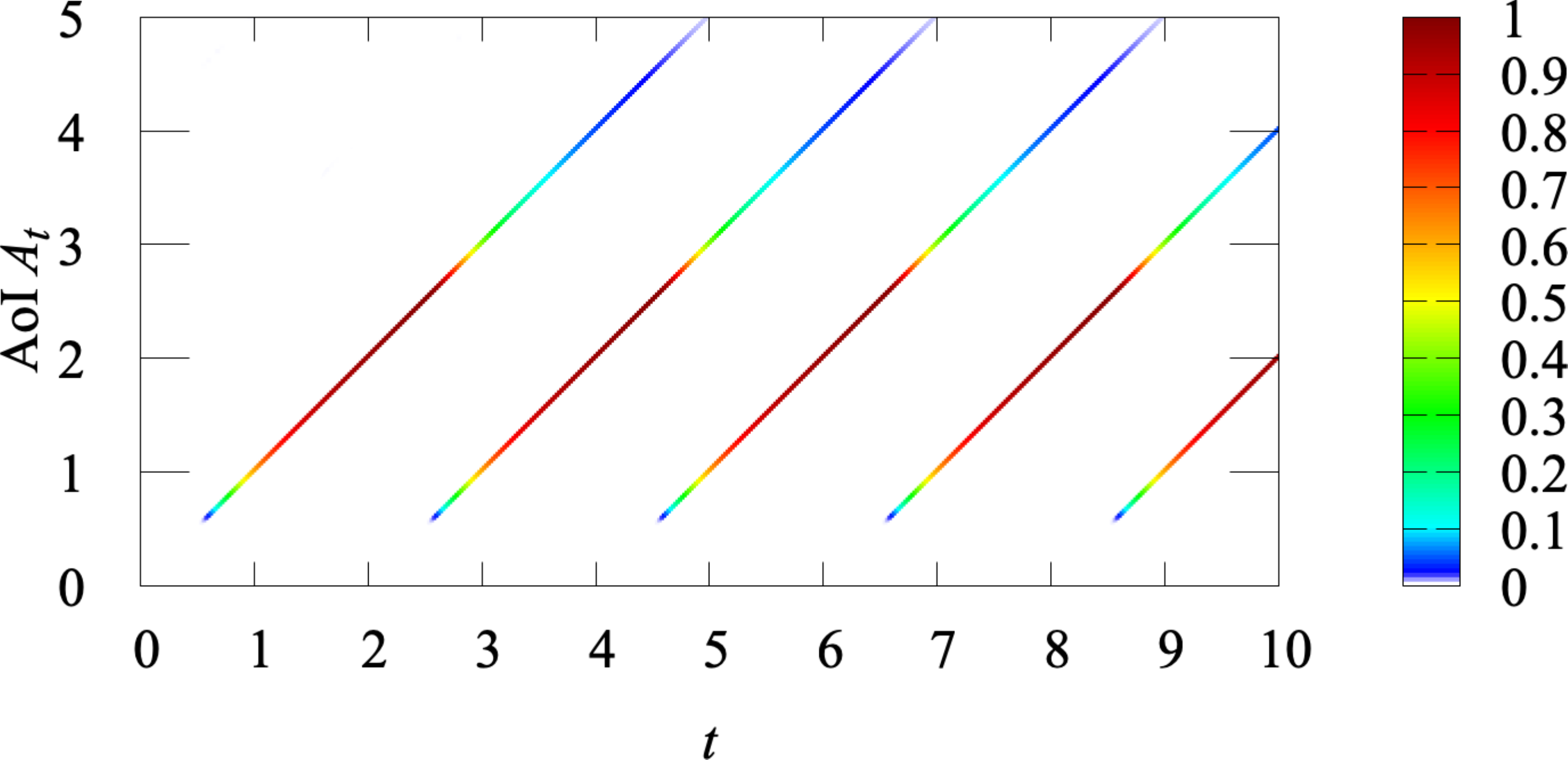} 
\subcaption{$c=0.1$, $\tau=2.0$}
\end{minipage}
\begin{minipage}[b]{0.49\linewidth}
\includegraphics[scale=0.17]{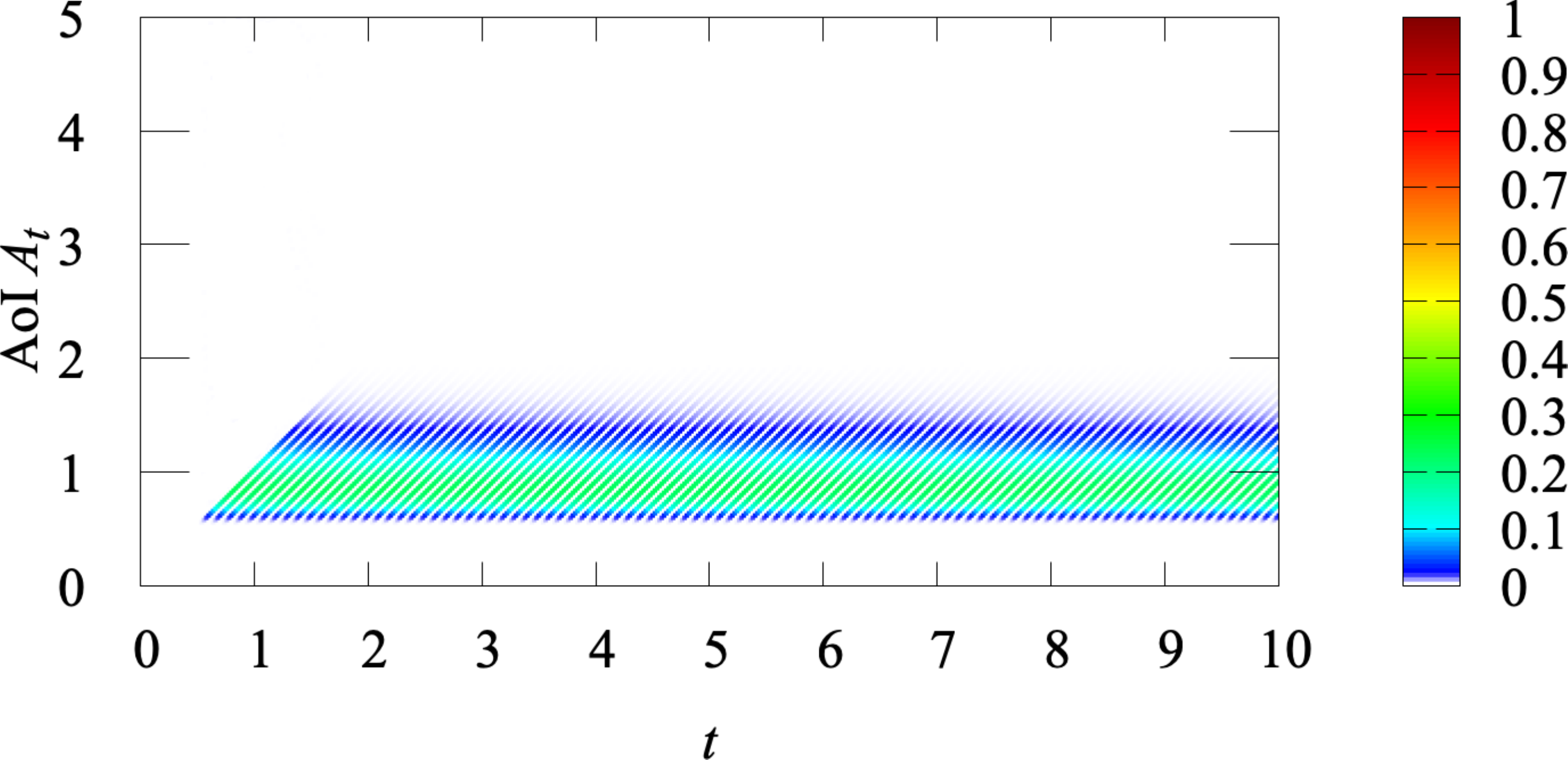} 
\subcaption{$c=0.1$, $\tau=0.1$}
\end{minipage}
\begin{minipage}[b]{0.49\linewidth}
\includegraphics[scale=0.17]{figs/r12_1hmpmf20b.pdf} 
\subcaption{$c=10.0$, $\tau=2.0$}
\end{minipage}
\begin{minipage}[b]{0.49\linewidth}
\includegraphics[scale=0.17]{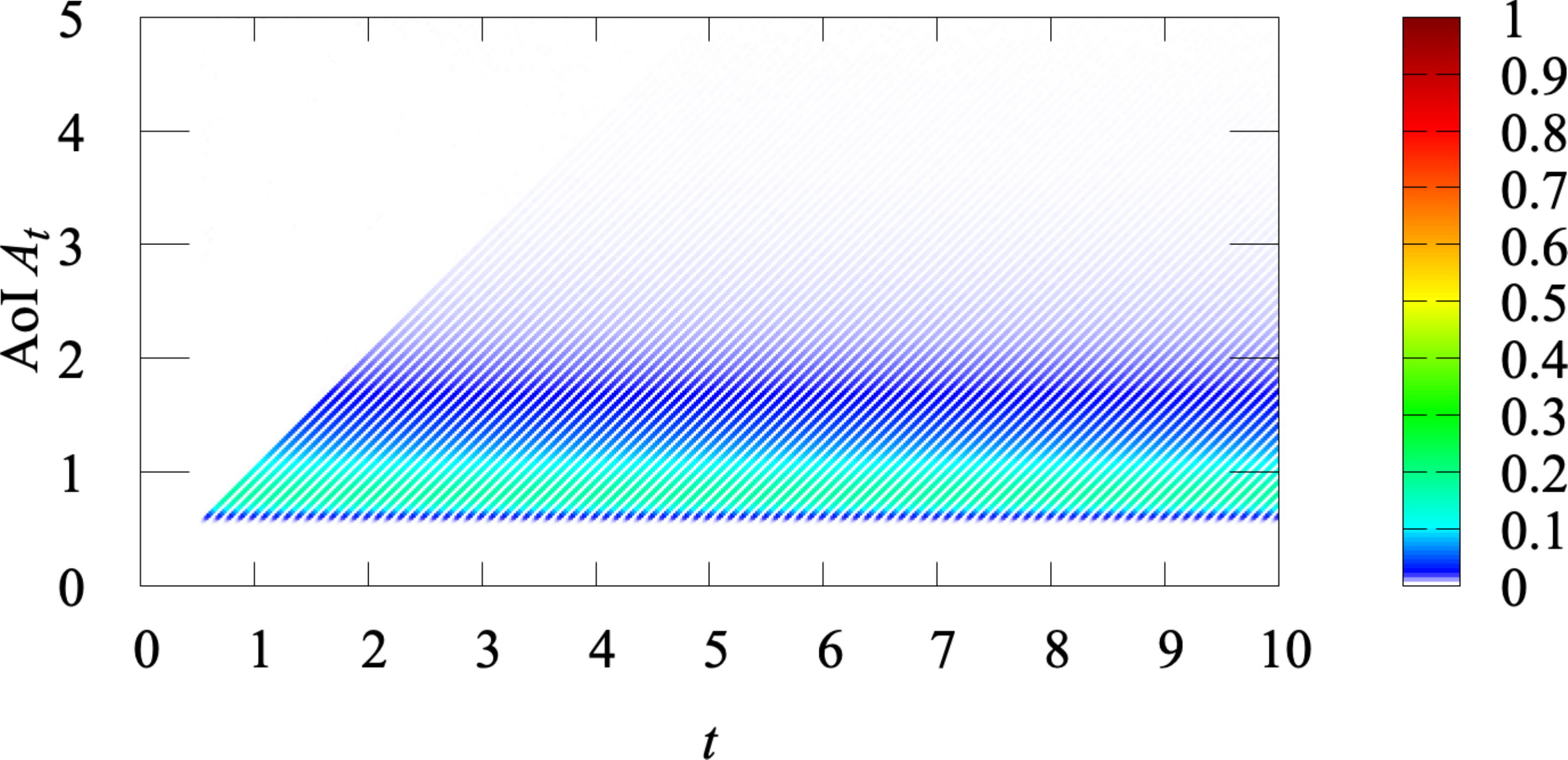} 
\subcaption{$c=10.0$, $\tau=0.1$}
\end{minipage}
\caption{Heatmap of the probability mass function of the AoI distribution in Example 2 ($s=0.75$).}\label{fig:tau_vs_heatmap3_1_B}
\end{figure*}

\begin{figure*}[ht!]
\centering
\begin{minipage}[b]{0.49\linewidth}
\includegraphics[scale=0.17]{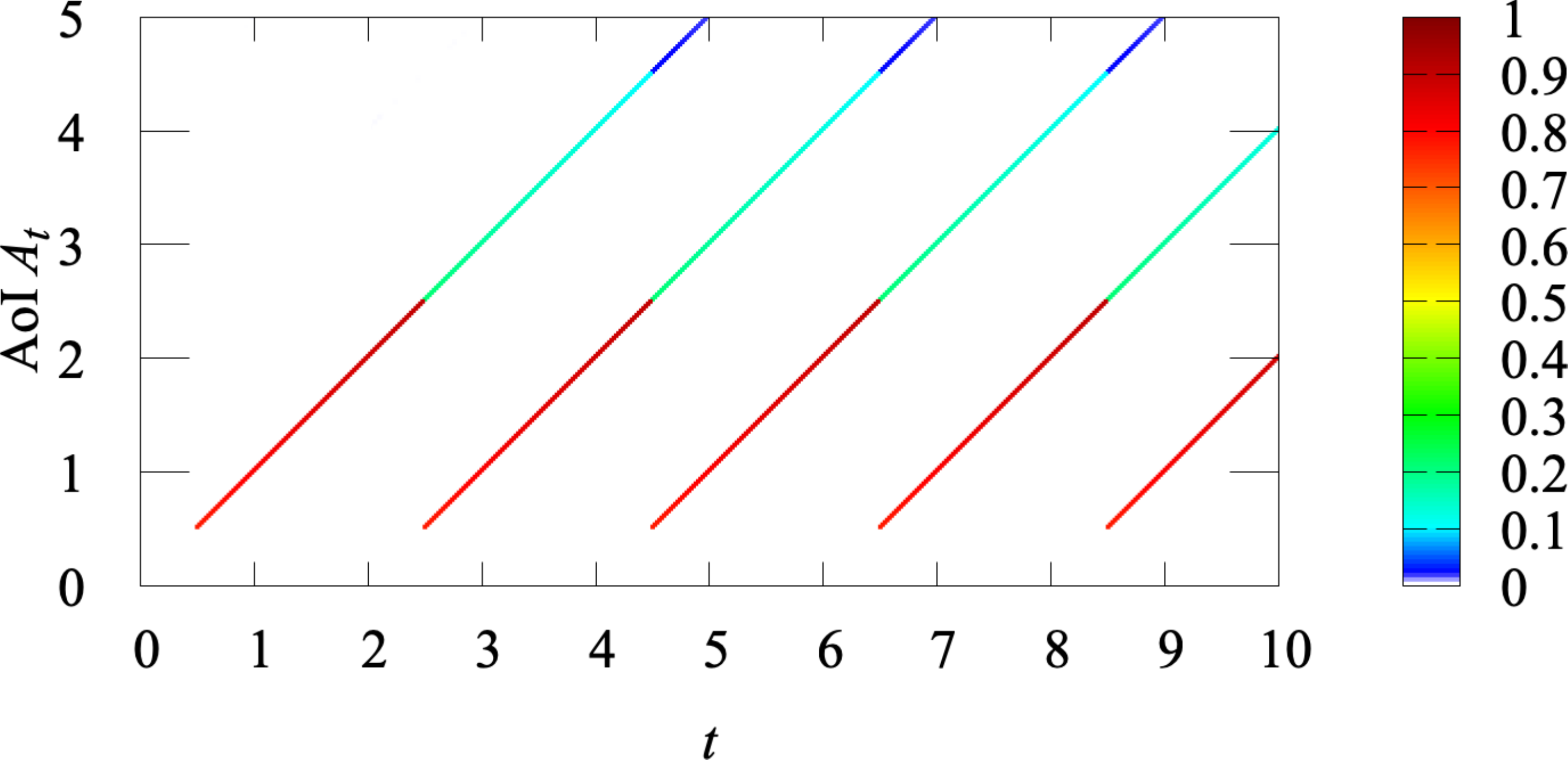} 
\subcaption{$c=0.1$, $\tau=2.0$}
\end{minipage}
\begin{minipage}[b]{0.49\linewidth}
\includegraphics[scale=0.17]{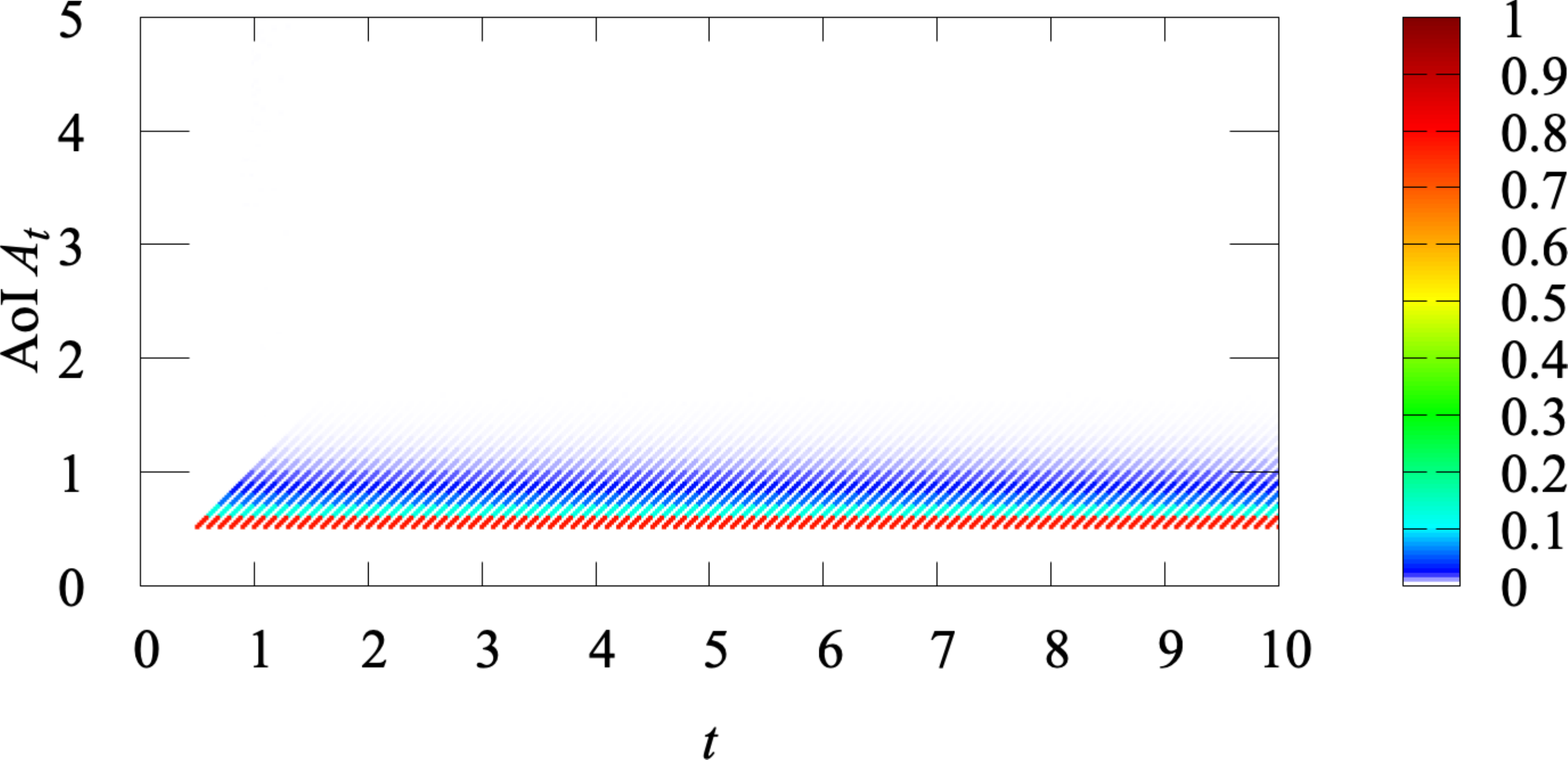} 
\subcaption{$c=0.1$, $\tau=0.1$}
\end{minipage}
\begin{minipage}[b]{0.49\linewidth}
\includegraphics[scale=0.17]{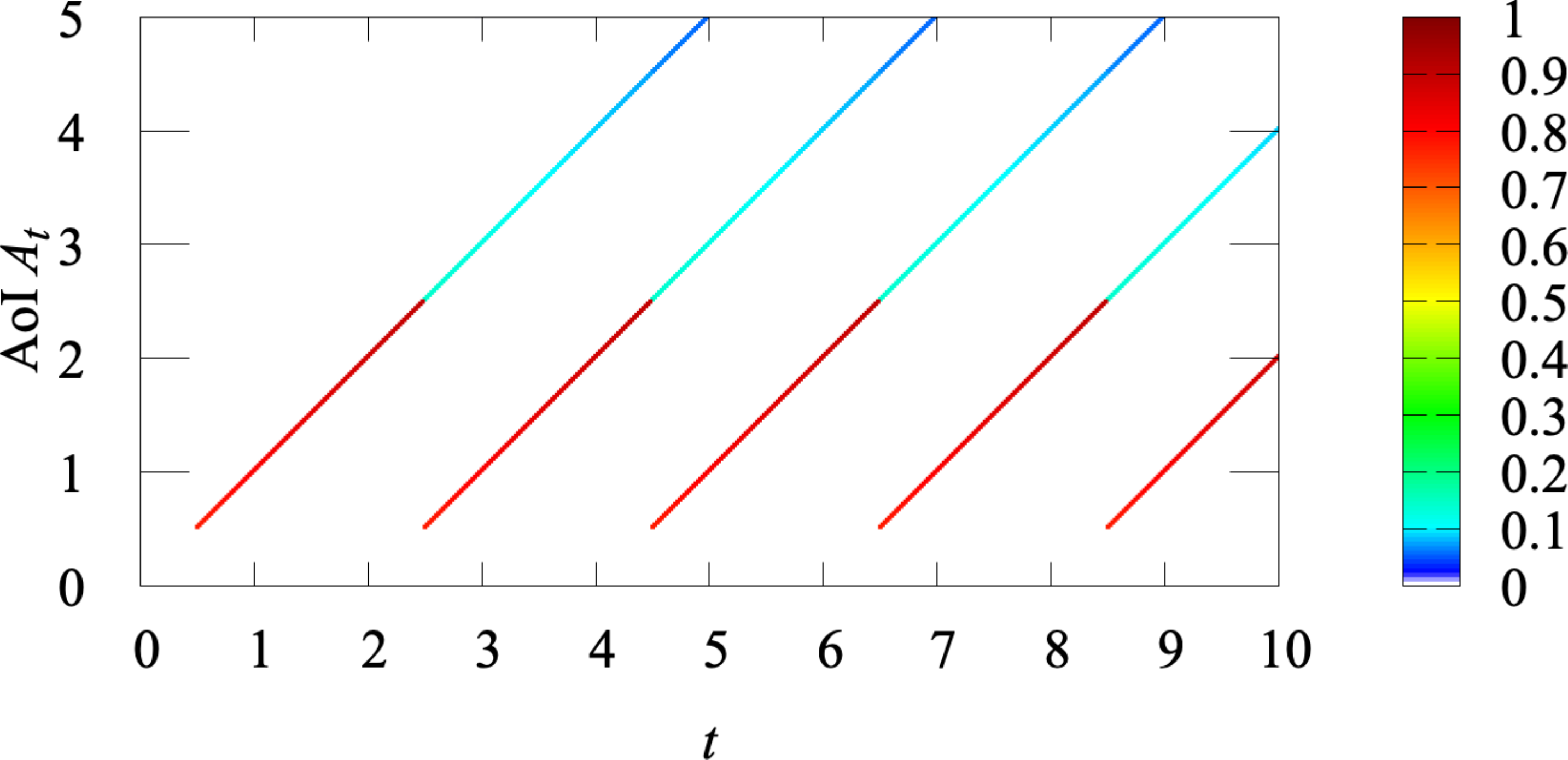} 
\subcaption{$c=10.0$, $\tau=2.0$}
\end{minipage}
\begin{minipage}[b]{0.49\linewidth}
\includegraphics[scale=0.17]{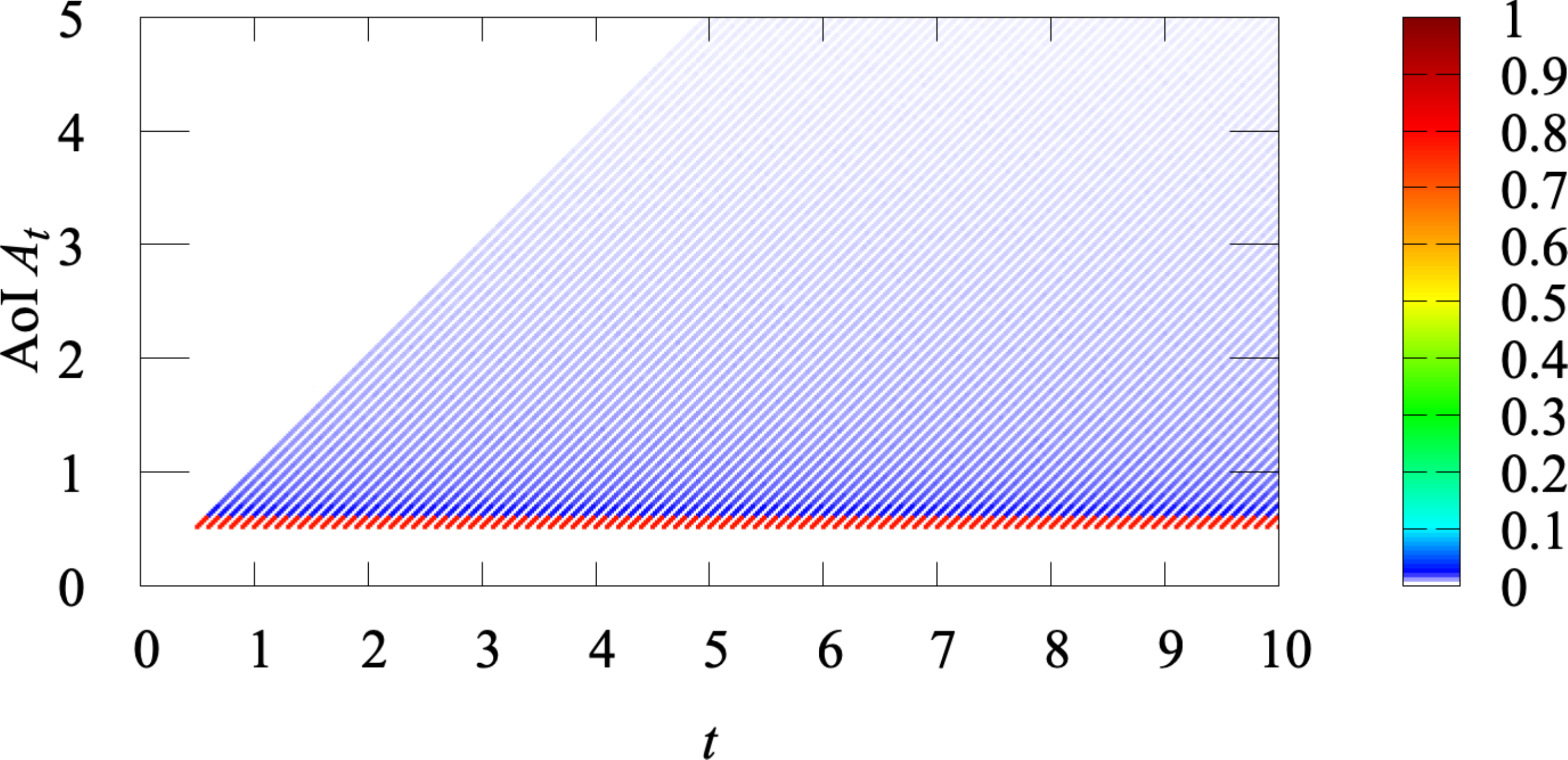} 
\subcaption{$c=10.0$, $\tau=0.1$}
\end{minipage}
\caption{Heatmap of the probability mass function of the AoI distribution in Example 1 ($s=1.25$).}\label{fig:tau_vs_heatmap3_2_A}
\end{figure*}
\begin{figure*}[ht!]
\centering
\begin{minipage}[b]{0.49\linewidth}
\includegraphics[scale=0.17]{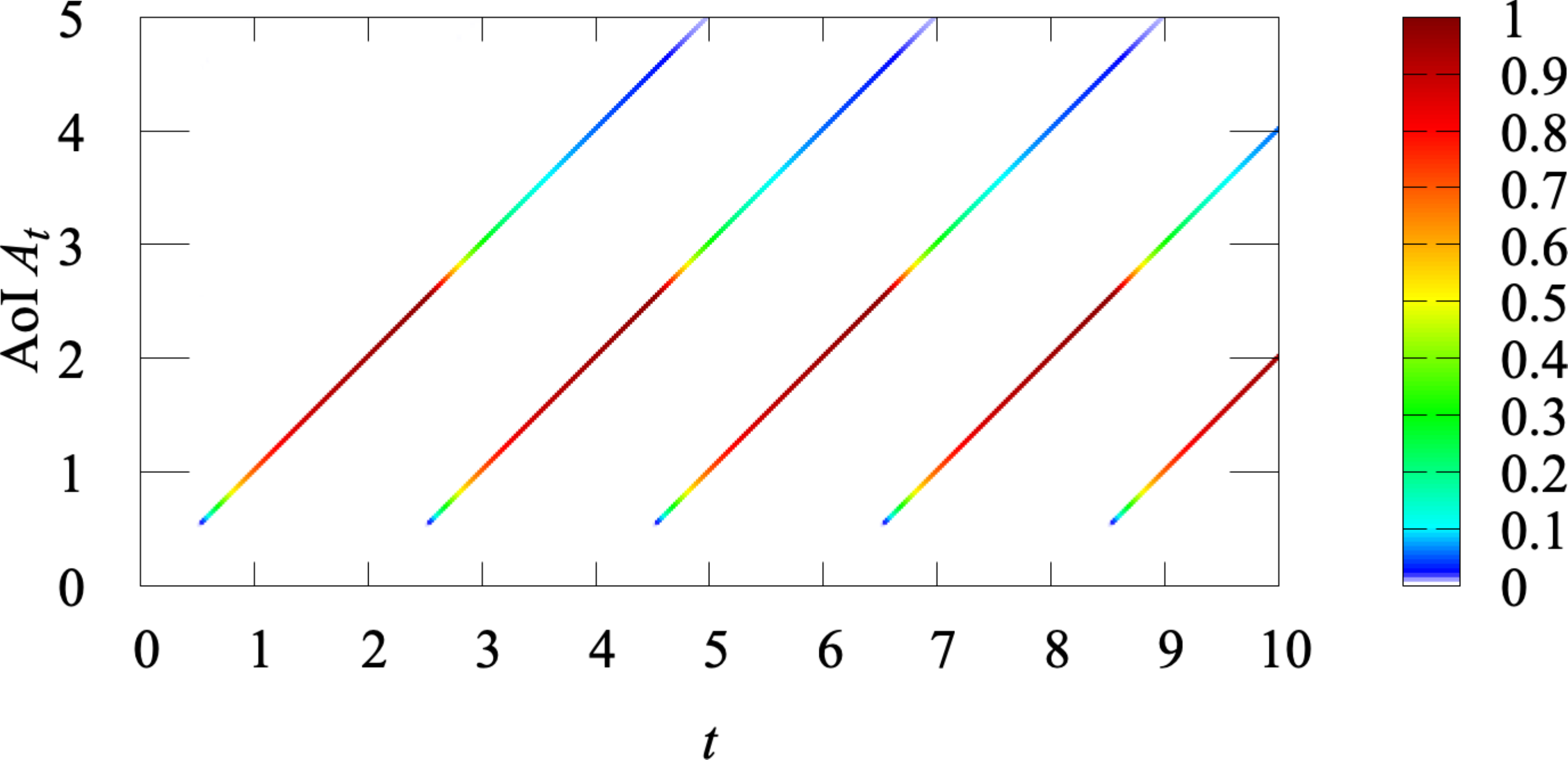} 
\subcaption{$c=0.1$, $\tau=2.0$}
\end{minipage}
\begin{minipage}[b]{0.49\linewidth}
\includegraphics[scale=0.17]{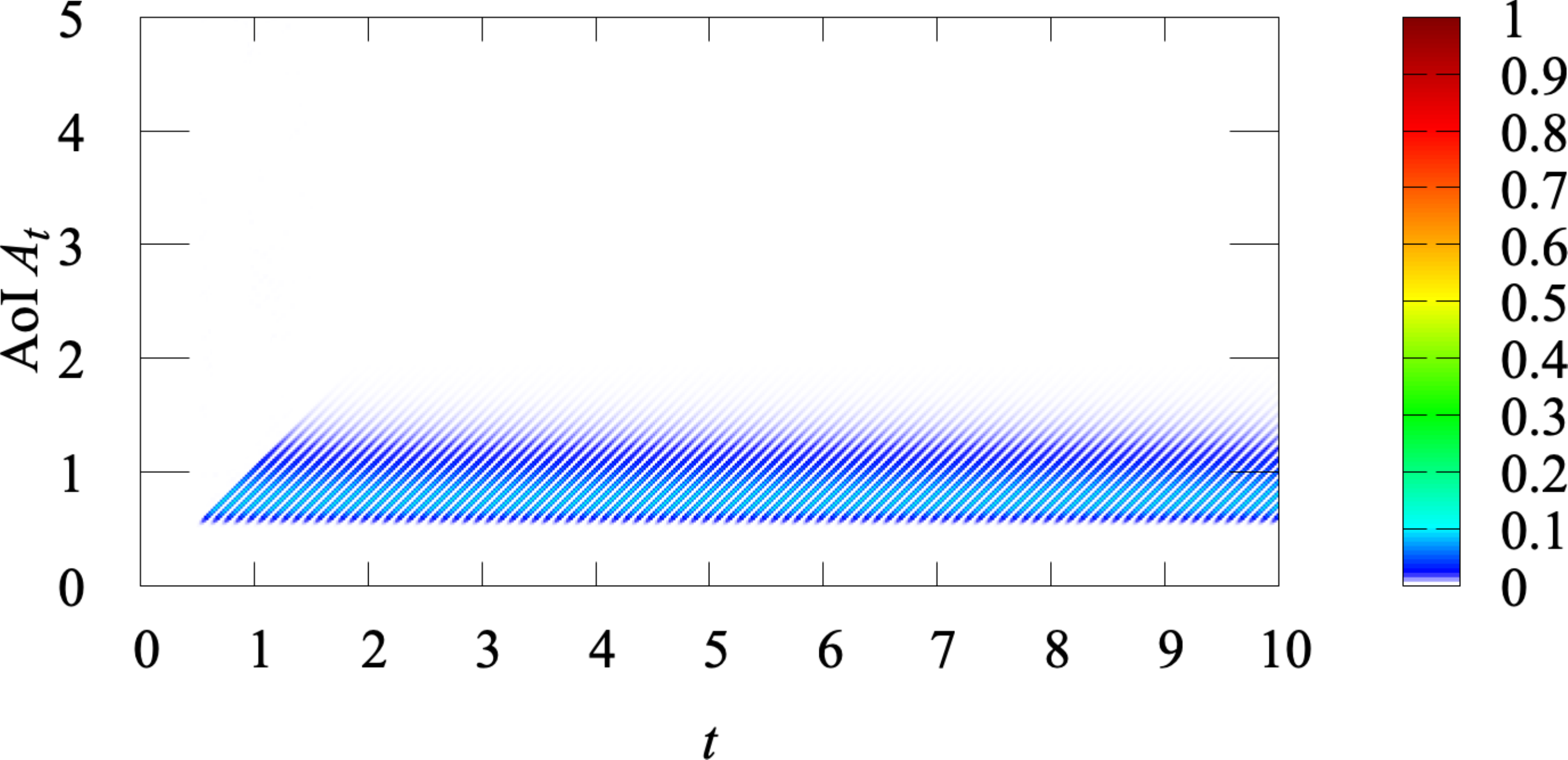} 
\subcaption{$c=0.1$, $\tau=0.1$}
\end{minipage}
\mbox{}
\\
\begin{minipage}[b]{0.49\linewidth}
\includegraphics[scale=0.17]{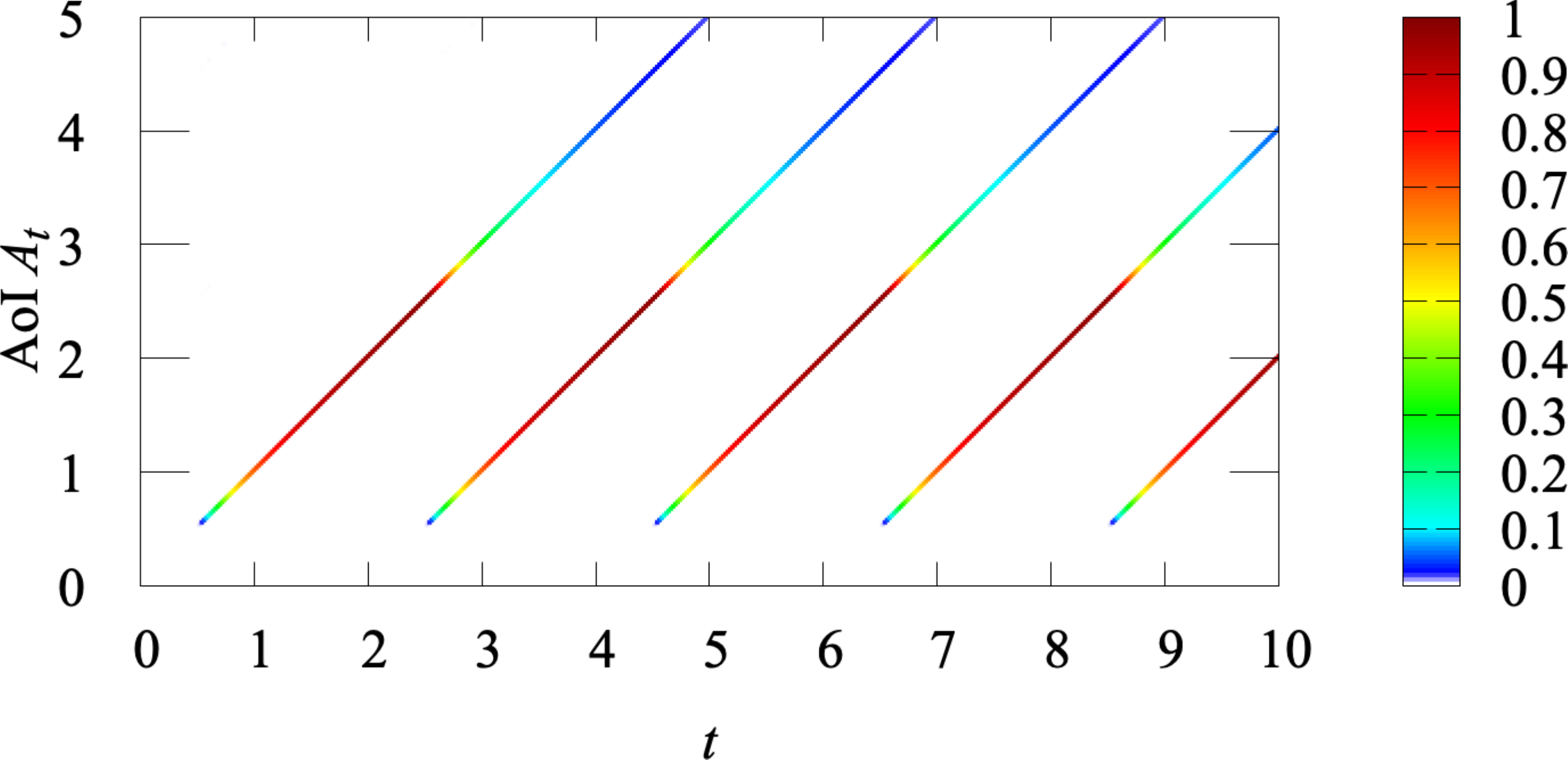} 
\subcaption{$c=10.0$, $\tau=2.0$}
\end{minipage}
\begin{minipage}[b]{0.49\linewidth}
\includegraphics[scale=0.17]{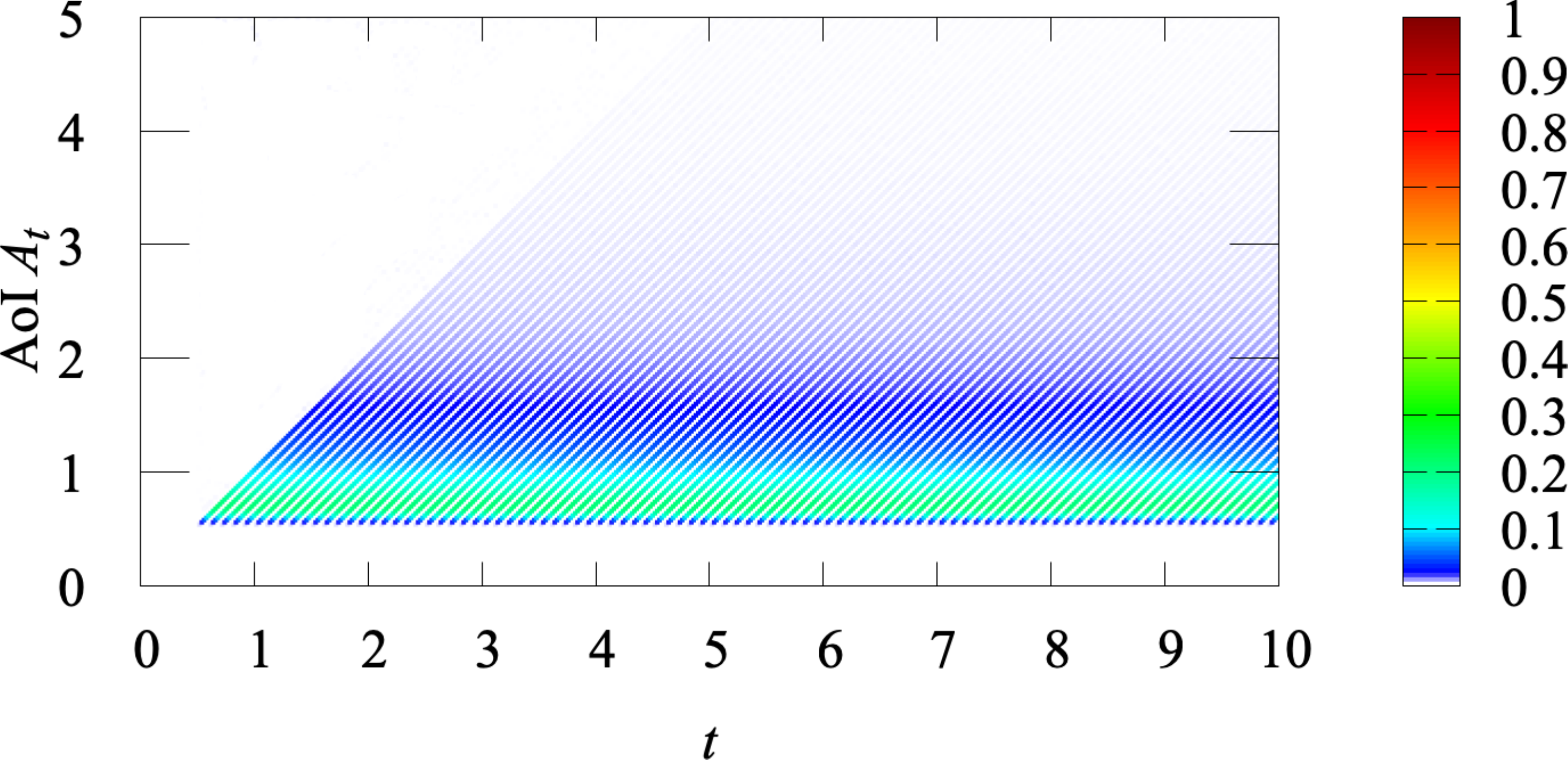} 
\subcaption{$c=10.0$, $\tau=0.1$}
\end{minipage}
\caption{Heatmap of the probability mass function of the AoI distribution in Example 2 ($s=1.25$).}\label{fig:tau_vs_heatmap3_2_B}
\end{figure*}
\begin{figure*}[p]
\centering
\begin{minipage}[b]{0.49\linewidth}
\includegraphics[scale=0.3]{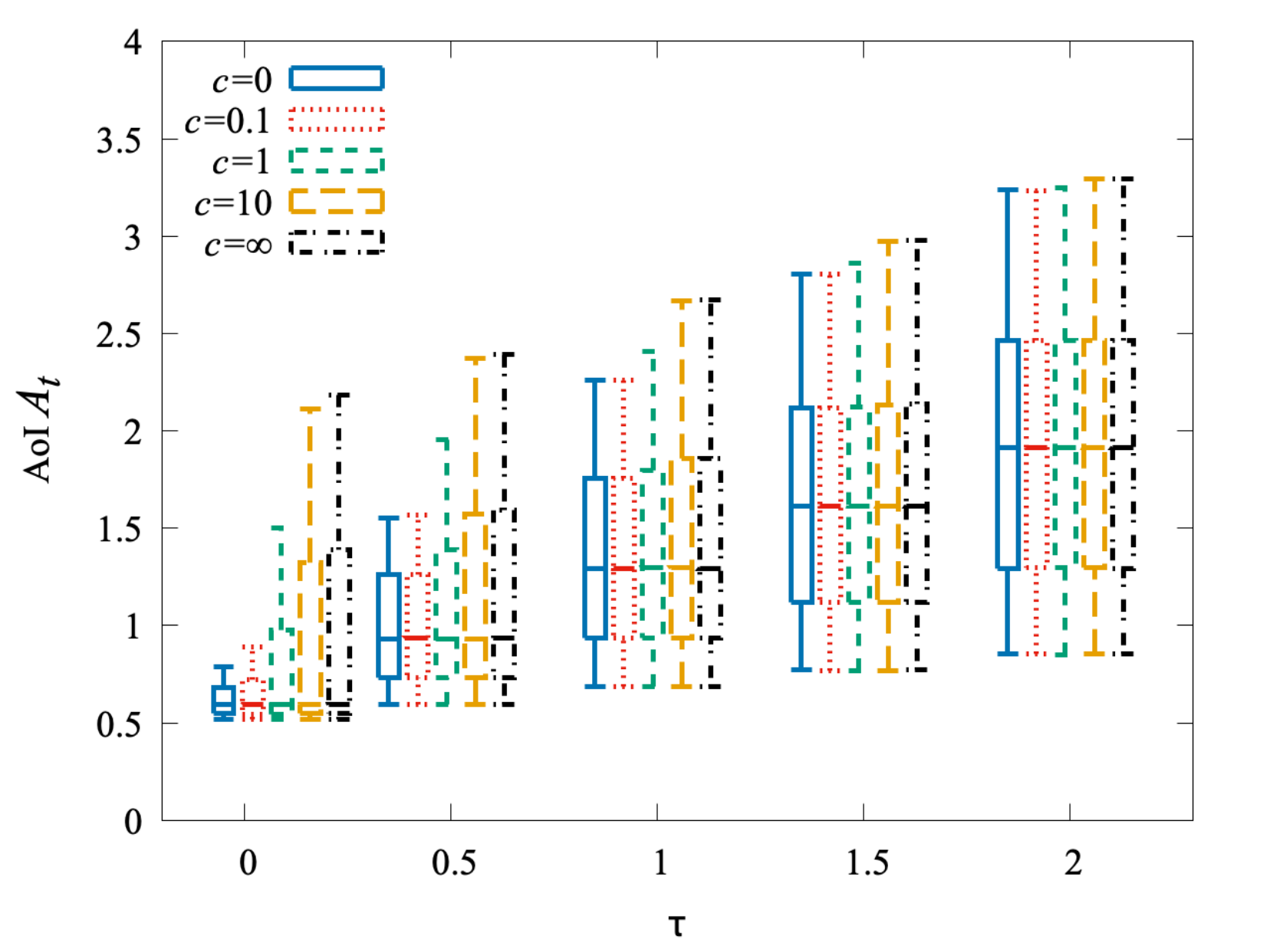} 
\subcaption{$s=0.75$}
\end{minipage}
\begin{minipage}[b]{0.49\linewidth}
\includegraphics[scale=0.3]{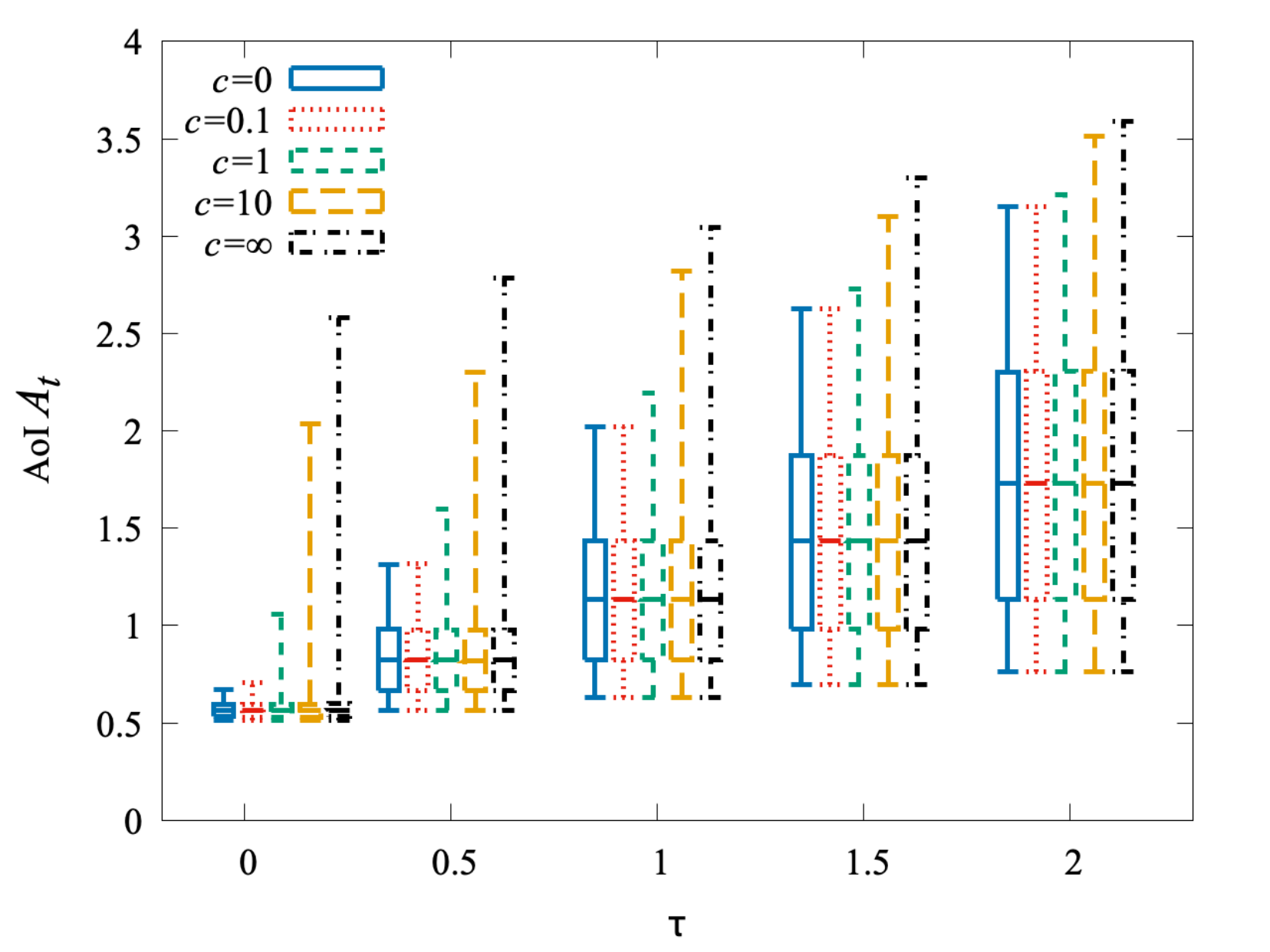} 
\subcaption{$s=1.25$}
\end{minipage}
\caption{Time-averaged distribution of the AoI for packet generation period $\tau$ in Example 1.}\label{fig:tau_vs_aoi3}
\mbox{}
\\
\centering
\begin{minipage}[b]{0.49\linewidth}
\includegraphics[scale=0.3]{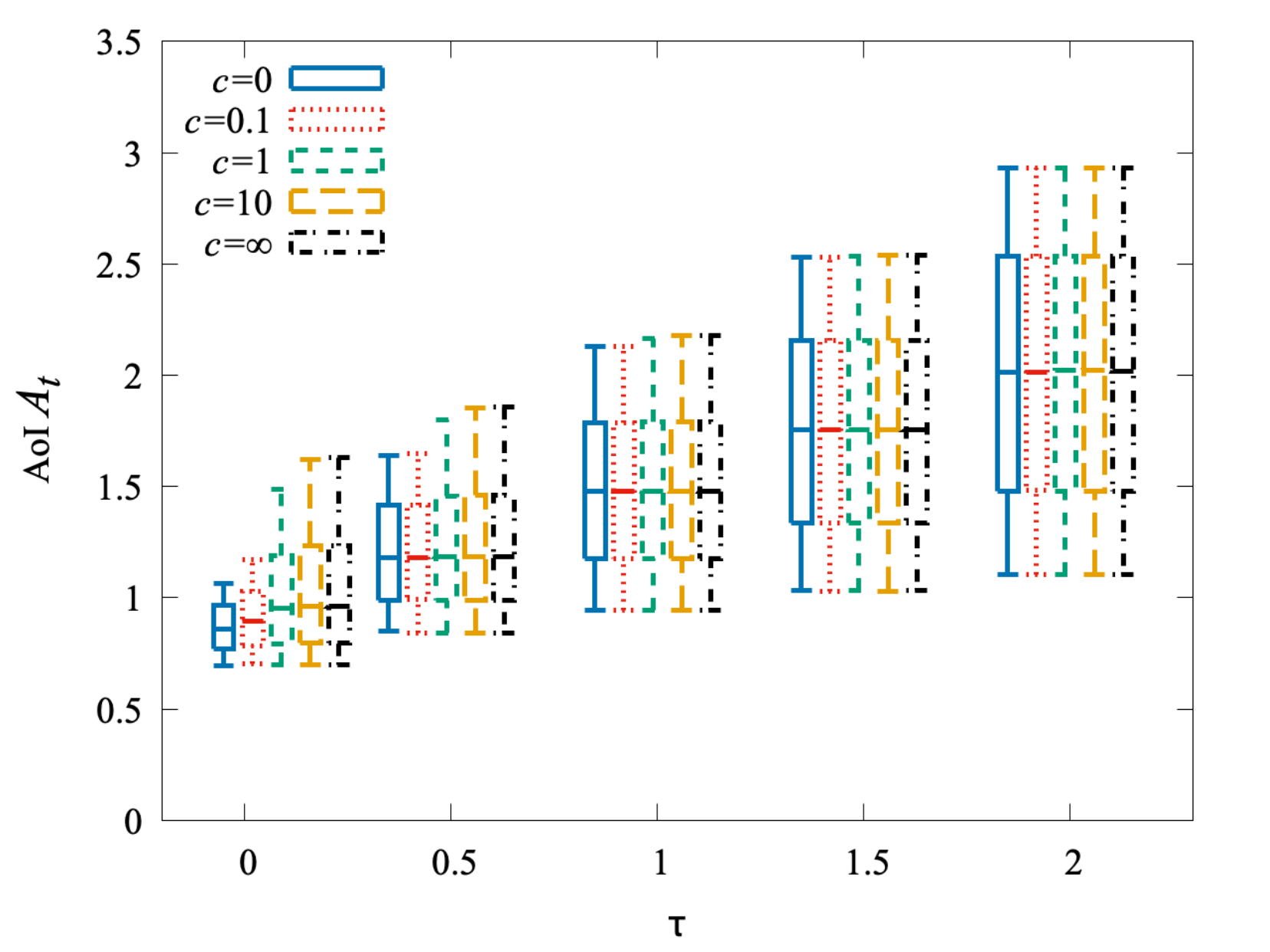} 
\subcaption{$s=0.75$}
\end{minipage}
\begin{minipage}[b]{0.49\linewidth}
\includegraphics[scale=0.3]{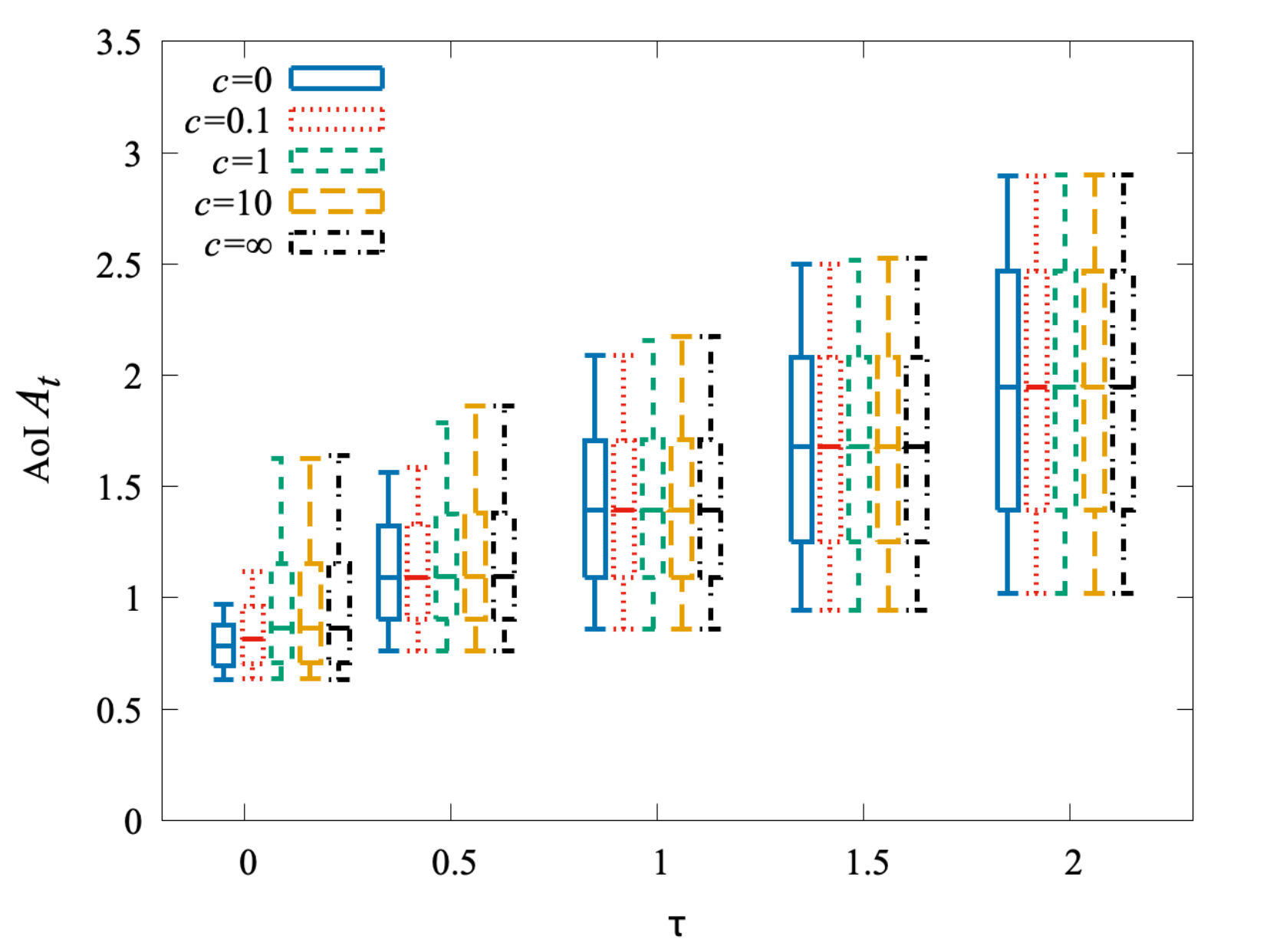} 
\subcaption{$s=1.25$}
\end{minipage}
\caption{Time-averaged distribution of the AoI for packet generation period $\tau$ in Example 2.}\label{fig:tau_vs_aoi4}
\mbox{}
\\
\centering
\begin{minipage}[b]{0.485\linewidth}
\includegraphics[scale=0.3]{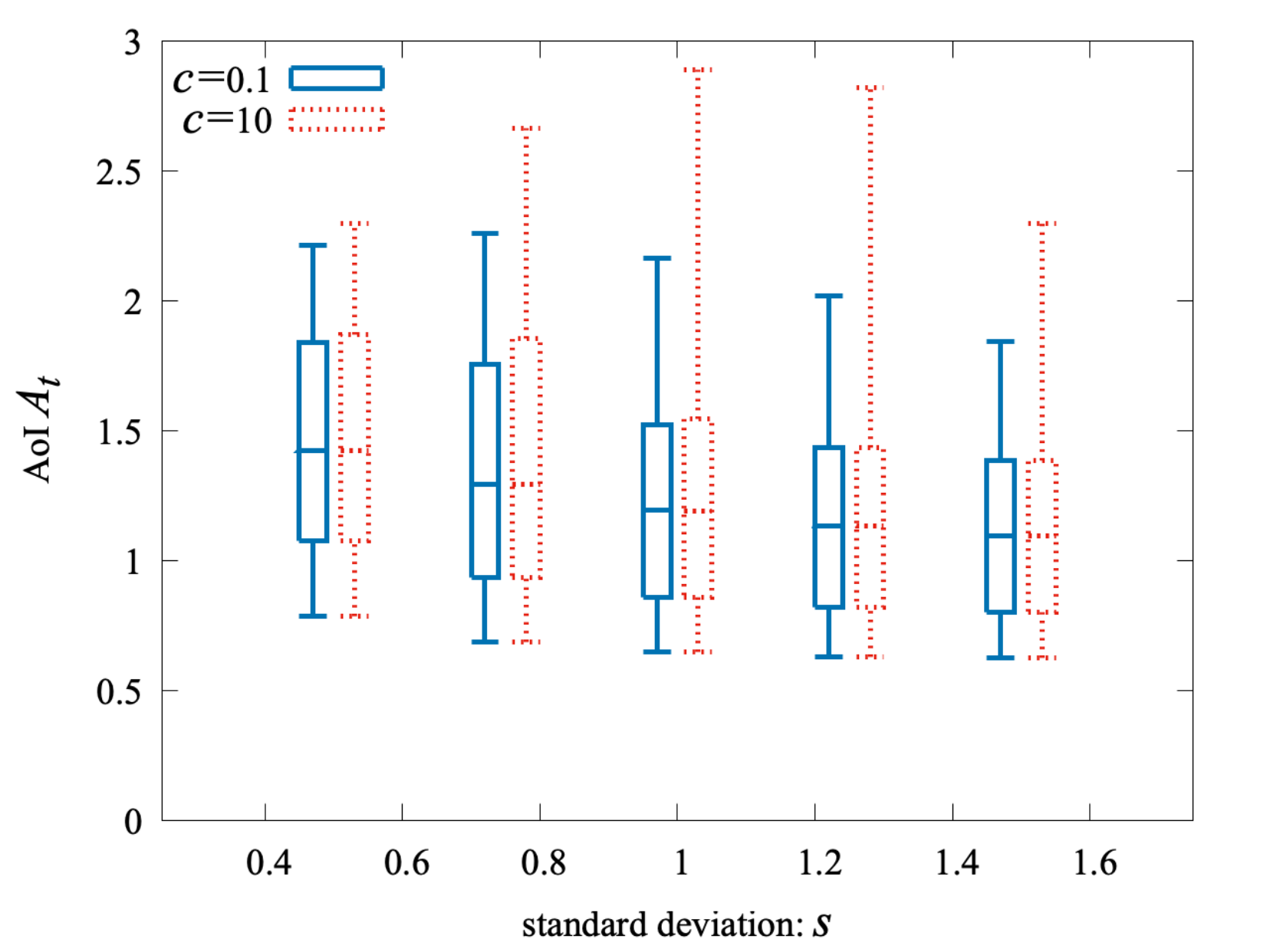} 
\subcaption{Example 1}
\end{minipage}
\begin{minipage}[b]{0.485\linewidth}
\includegraphics[scale=0.30]{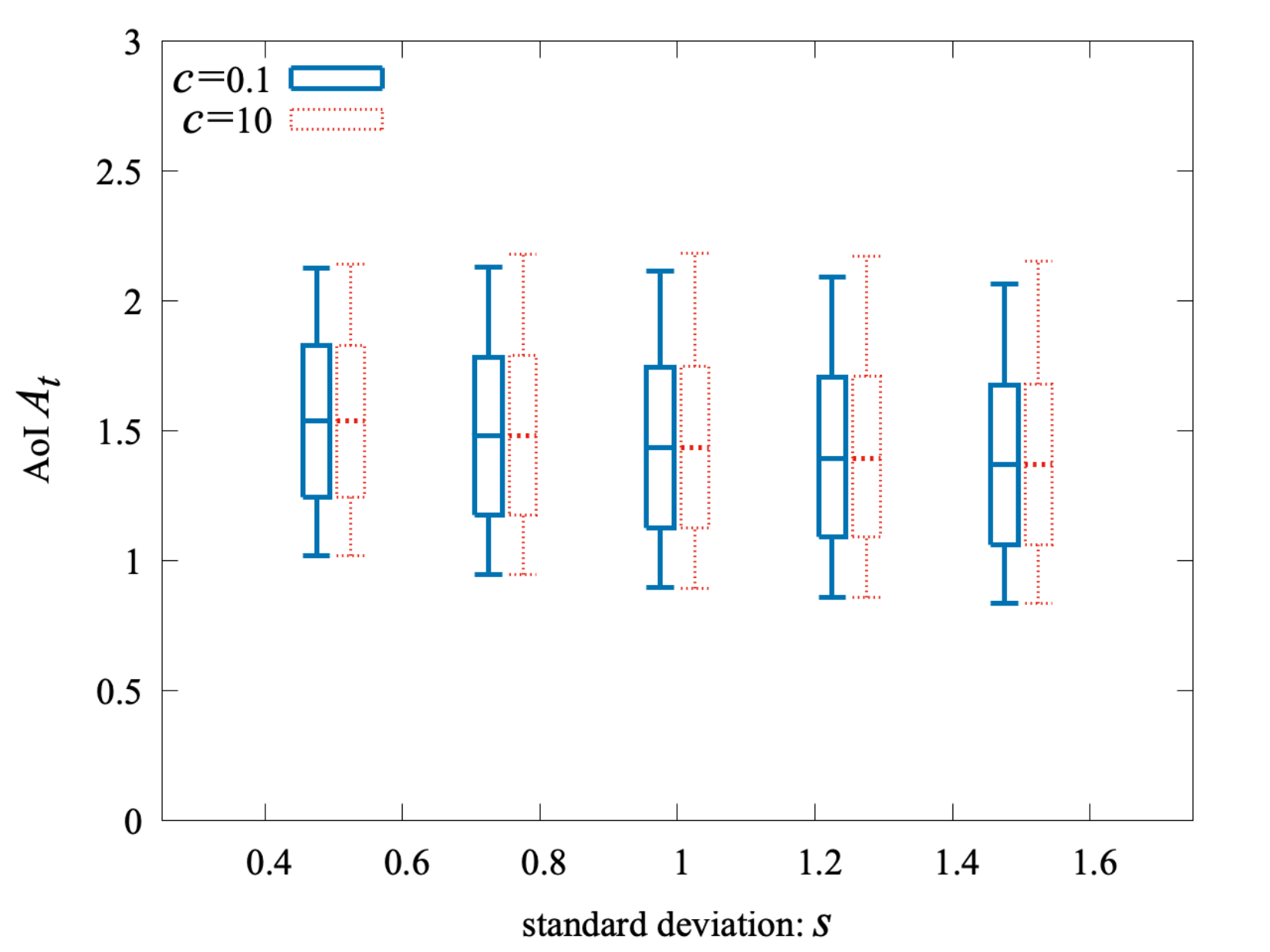} 
\subcaption{Example 2}
\end{minipage}
\caption{Time-averaged distribution of the AoI for standard deviation $s$. }\label{fig:std_vs_aoi2}
\end{figure*}

\section{Numerical Experiments}\label{sec:num}
In this section, we investigate the effect of the autocorrelation of
the virtual delay process $(X_t)_{t \ge 0}$ on the AoI through numerical experiments.
We focus on the case that $(Z_t)_{t \geq 0}$ follows an OU process,
and we consider two choices for the function $g$, as presented in
Examples 1 and 2.

The parameter values are set as follows: we fix the mean $\mu = 1$ and
the lower bound $x_{\min} = 0.5$ of the process $(X_t)_{t \ge 0}$, and
determine the parameters $\hat{\mu}$ and $\hat{s}$ of the function $g$
so that the standard deviation $s := \sqrt{\sigma(0)}$ of $(X_t)_{t
\ge 0}$ is equal to one of the values in $\{0.75, 1.25\}$.
Additionally, the value of $\kappa$ in the OU process is chosen to
satisfy $\sigma(c)/\sigma(0) = e^{-1}$ for a given time constant $c$. 

Figs.~\ref{fig:aoi_sim:ex1} and \ref{fig:aoi_sim:ex2} show the AoI
simulation results and the corresponding theoretical probability mass function (heat maps) obtained by (\ref{eq:AoI-CCDF-Gauss}) for
the case where $s = 0.75$, $c = 10$, and $\tau = 2.0$. The parameter values
$(\hat{\mu}, \hat{s}, \kappa)$ used in this setting are $(0.452,\
1.312,\ 0.081)$ in Example~1 and $(-1.282,\ 1.085,\ 0.066)$ in
Example~2.
In the left panels of Figs.~\ref{fig:aoi_sim:ex1} and
\ref{fig:aoi_sim:ex2}, we show overlaid sample paths from 500
simulations using the above parameter values. The right panels display
heat maps of the probability mass function, where the plotted value
corresponds to $\Pr(A_t > x) - \Pr(A_t > x + \delta)$, with $t$ on the
horizontal axis, $x$ on the vertical axis, and $\delta$ denoting the
step size parameter. In our numerical experiments, we fix $\delta=0.02$. 

The results show that the probability mass function obtained via numerical calculation accurately captures the region through which the AoI sample paths evolve. We observe that $A_t$ exhibits periodic behavior in $t$ with period $\tau = 2.0$. Moreover, the heat map of the probability mass function confirms that $\Pr(A_t > x)$ is a step function with respect to $x$.

We next examine the relationship between the packet generation
interval $\tau$ and the AoI distribution. 
Figs.~\ref{fig:tau_vs_heatmap3_1_A} and \ref{fig:tau_vs_heatmap3_1_B} show
heat maps of the probability mass function of the AoI distribution for
$s = 0.75$, while Figs.~\ref{fig:tau_vs_heatmap3_2_A} and \ref{fig:tau_vs_heatmap3_2_B} 
present the corresponding results for $s = 1.25$.
From
Figs.~\ref{fig:tau_vs_heatmap3_1_A} and \ref{fig:tau_vs_heatmap3_1_B}, we
observe that the rate of decay in
the tail probability of the AoI with decreasing $\tau$ becomes more
gradual as the autocorrelation of the virtual delay process increases
(i.e., as the time constant $c$ increases). A similar trend is
observed in Figs.~\ref{fig:tau_vs_heatmap3_2_A} and \ref{fig:tau_vs_heatmap3_2_B}.

Figs.~\ref{fig:tau_vs_aoi3} and \ref{fig:tau_vs_aoi4} 
show box plots of the AoI values corresponding to the 10th, 25th,
50th, 75th, and 90th percentiles (from bottom to top) for $s = 0.75$
and $s = 1.25$. These percentiles are computed based on the
time-averaged distribution of the AoI, given by $\frac{1}{\tau}
\int_x^{x+\tau} \Pr(A_t > x) \, \mathrm{d}t$, for time constants 
$c \in \{0, 0.1, 1, 10, \infty\}$ and packet generation intervals 
$\tau \in \{0.1, 0.5, 1.0, 1.5, 2.0\}$. 
The two extreme cases $c=0$ and $c=\infty$ are computed directly using (\ref{eq:AoI-CCDF-Gauss}): 
for $c=0$, it follows that $Z_0,Z_{\tau},...$ are i.i.d., while for $c=\infty$, we have $Z_{i\tau}=Z_0$ ($i=0,1,\ldots$).

As shown in Figs.~\ref{fig:tau_vs_aoi3} and \ref{fig:tau_vs_aoi4}, we observe that
(consistent with the trends in Figs.~\ref{fig:tau_vs_heatmap3_1_A}--\ref{fig:tau_vs_heatmap3_1_B})
the tail probabilities of the AoI decay more gradually as the time
constant $c$ increases. This effect is particularly pronounced when
$\tau$ is small, where the influence of $c$ on the AoI distribution is
clearly evident.
Conversely, when $\tau$ is large, the impact of $c$ on the AoI
distribution becomes negligible. In this regime, the delays $(D_n)_{n=0,1,\ldots}$ behave
approximately as independent and identically distributed random
variables. We also observe that the AoI distributions for $c = 0.1$
and $c = 10$ closely resemble those for $c = 0$ and $c = \infty$,
respectively.

The relationship between the standard deviation $s$ of the virtual
delay process $(X_t)_{t \ge 0}$ and the AoI distribution for $c\in\{0.1,10\}$, is
illustrated in Fig.~\ref{fig:std_vs_aoi2}. This figure shows that
when the autocorrelation of the virtual delay process is small, the
AoI tends to decrease as the standard deviation $s$ increases.
On the other hand, when the autocorrelation is large, the tail
probability of the AoI distribution decays more slowly. In such cases,
increasing the standard deviation does not necessarily lead to lower
AoI. Additionally, for a fixed value of $s$, we observe that higher
autocorrelation in the virtual delay process consistently results in
higher AoI.

\section{Conclusion}\label{sec:conc}
In this paper, we introduced a framework in which the sequence of
delays is constructed from a non-negative continuous-time stochastic
process, referred to as the virtual delay process, as a theoretical
model for analyzing the AoI.

We first derived an expression for the transient probability distribution of the AoI. Based on this expression, we applied stochastic ordering techniques to show that stronger dependence in the sequence of delays leads to degraded AoI performance. As a special case, we considered virtual delay processes composed from stationary Gaussian processes and developed a numerical method for computing the AoI distribution. Finally, we investigated the behavior of the model through numerical experiments.

As future work, we plan to explore efficient methods for computing the AoI distribution when $(Z_t)_{t \geq 0}$ follows a general Gaussian process. Further analysis of the autocorrelation structure of the AoI process will also be an important direction for study.

\appendices
\section{Computation of the AoI distribution for virtual delay
processes composed from the OU process}\label{appendix:OU-compute}

We define a conditional complementary CDF $\overline{\Psi}_n(y \mid a_0, a_1, \ldots, a_{n-1})$ 
and the corresponding density $\psi_n(y \mid a_0, a_1,\ldots, a_{n-1})$ as follows:
\begin{align*}
\overline{\Psi}_n(y \mid a_0, a_1, \ldots, a_{n-1})
=
\Pr\left(
Z_{n\tau} > y \mid \bigcap_{i=1}^{n}\{Z_{(n-i)\tau} > a_{n-i}\}
\right),
\end{align*}
\[
\psi_n(y \mid a_0, a_1,\ldots, a_{n-1})
:=
-\frac{\partial}{\partial y}\overline{\Psi}_n(y \mid a_0,a_1,\ldots, a_{n-1}), 
\]
By definition, we have 
\begin{align}
\Pr(Z_{n\tau} > a_n, Z_{(n-1)\tau} > a_{n-1}, \ldots, Z_0 > a_0)
&=
\Pr(Z_0 > a_0)
\prod_{i=1}^n \Pr\left(Z_{i\tau} > a_i \mid \bigcap_{j=1}^{i}\{Z_{(i-j)\tau} > a_{i-j}\}\right)
\nonumber
\\
&=
\Pr(Z_0 > a_0) 
\prod_{i=1}^n 
\overline{\Psi}_i(a_i \mid a_0, a_1, \ldots, a_{i-1}). 
\label{eq:joint_prob_OU}
\end{align}
From the Markov property of OU processes 
and (\ref{eq:OU-transient-density}), we have
\begin{align*}
\lefteqn{\overline{\Psi}_n(y \mid a_0, a_1, \ldots, a_{n-1})}\\
&=
\frac{1}{\overline{\Psi}_n(a_{n-1} \mid a_0, a_1, \ldots, a_{n-2})}
\int_{a_{n-1}}^{\infty}
\Pr\Biggl(
Z_{n\tau} > y \mid \{Z_{(i-1)\tau} = u\}
\cap \Biggl(\bigcap^{i}_{j=2}\{Z_{(i-j)\tau}> a_{i-j}\} \Biggr)
\Biggr)
\psi_n(u \mid a_0, a_1, \ldots, a_{n-2}) {\rm d}u
\\
&=
\frac{1}{\overline{\Psi}_n(a_{n-1} \mid a_0, a_1, \ldots, a_{n-2})}
\int_{a_{n-1}}^{\infty}
\Pr(Z_{n\tau} > y \mid Z_{(i-1)\tau} = u)
\psi_n(u \mid a_0, a_1, \ldots, a_{n-2}) {\rm d}u
\\
&=
\frac{1}{\overline{\Psi}_n(a_{n-1} \mid a_0, a_1, \ldots, a_{n-2})}
\int_{u=a_{n-1}}^{\infty}
\left(
\int_{x=y}^{\infty} 
f_{N(ze^{-\kappa u},1-e^{-2\kappa u})} {\rm d}x 
\right)
\psi_n(u \mid a_0, a_1, \ldots, a_{n-2}) {\rm d}u.
\end{align*}
We then obtain the following recursions for $\psi_n(y \mid a_0, a_1,\ldots, a_{n-1})$ and $\overline{\Psi}_n(y \mid a_0, a_1,\ldots, a_{n-1})$:

\begin{align*}
&\psi_1(y \mid a_0)
=
\int_{u=a_0}^{\infty}
f_{N(ze^{-\kappa u},1-e^{-2\kappa u})}(y)
\cdot
f_{N(0, 1)}(u) {\rm d}u,
\\
&\psi_n(y \mid a_0, a_1, \ldots, a_{n-1})
\\
&\qquad{}=
\frac{1}{\overline{\Psi}_n(a_{n-1} \mid a_0, a_1, \ldots, a_{n-2})}
\int_{u=a_{n-1}}^{\infty}
f_{N(ze^{-\kappa u},1-e^{-2\kappa u})}(y)
\psi_n(u \mid a_0, a_1, \ldots, a_{n-2}) {\rm d}u,
\;\;
n=2,3,\ldots,
\\
&\overline{\Psi}_n(y \mid a_0, a_1, \ldots, a_{n-1})
=
1 - \int_0^y \psi_n(y \mid a_0, a_1, \ldots, a_{n-1}),
\quad
n = 0,1,\ldots. 
\end{align*}

Therefore, using these recursions and (\ref{eq:joint_prob_OU}), 
the $(k_t-\theta_t(x))$-fold integral arising 
from (\ref{eq:AoI-CCDF-Gauss}) can be reduced to $2(k_t-\theta_t(x))$
simple integrals, which leads to a significant reduction in the amount
of computation.


\begin{thebibliography}{99}
\bibitem{Kaul11}
S.\ Kaul, M.\ Gruteser, V.\ Rai, and J.\ Kenney,
``Minimizing Age of Information in Vehicular Networks,''
in  \textit{Proc.\ IEEE SECON 2011}, pp.\ 350--358, 2011.

\bibitem{Kaul12-1}
S.\ Kaul, R.\ Yates, and M.\ Gruteser, 
``Real-Time Status: How Often Should One Update?,'' 
in \emph{Proc.\ IEEE INFOCOM 2012}, pp.\ 2731--2735, 2012.

\bibitem{Kosta2017}
A.\ Kosta, N.\ Pappas, and V.\ Angelakis, 
``Age of Information: A New Concept, Metric, and Tool,'' 
\emph{Foundations and Trends\textcircled{\tiny R} in Networking}, vol.\ 12, no.\ 3, 
pp.\ 162--259, 2017.

\bibitem{Yates2021}
R.\ D.\ Yates et al., ``Age of Information: An Introduction and Survey,'' {\it IEEE J.\ Sel.\ Areas Commun.}, vol.\ 39, no.\ 5, pp.\ 1183--1210, 2021.

%

\bibitem{Costa16}
M.\ Costa, M.\ Codreanu, and A.\ Ephremides, 
``On the Age of Information in Status Update Systems with Packet Management,'' 
\emph{IEEE Trans.\ Inf.\ Theory}, vol.\ 62, no.\ 4, pp.\ 1897--1910, 2016.

\bibitem{Kam18}
C.\ Kam, S.\ Kompella, G.\ D.\ Nguyen, J.\ E.\ Wieselthier, and A.\ Ephremides, 
``On the Age of Information with Packet Deadlines,'' 
\emph{IEEE Trans.\ Inf.\ Theory}, 
vol.\ 64, no.\ 9, pp.\ 6419--6428, 2018.

\bibitem{Bedewy19-1}
A.\ M.\ Bedewy, Y.\ Sun, and N.\ B.\ Shroff, 
``Minimizing the Age of Information Through Queues,'' 
\emph{IEEE Trans.\ on Inf.\ Theory}, vol.\ 65, no.\ 8, 
pp.\ 5215--5232, 2019.

\bibitem{Bedewy19-2}
A.\ M.\ Bedewy, Y.\ Sun, and N.\ B.\ Shroff, 
``The Age of Information in Multihop Networks,'' 
\emph{IEEE/ACM Trans.\ Netw.}, vol.\ 27, no.\ 3, pp.\ 1248--1257, 2019.

\bibitem{Inoue19}
Y.\ Inoue, H.\ Masuyama, T.\ Takine, and T.\ Tanaka, 
``A General Formula for the Stationary Distribution of the Age of 
Information and Its Application to Single-Server Queues,''
\emph{IEEE Trans.\ Inf.\ Theory}, vol.\ 65, no.\ 12, pp.\ 8305--8324, 2019.

\bibitem{Yates19}
R.\ D.\ Yates and S.\ K.\ Kaul, 
``The Age of Information: Real-Time Status Updating by Multiple
Sources,'' 
\emph{IEEE Trans.\ Inf.\ Theory}, vol.\ 65, no.\ 3, 
pp.\ 1807--1827, 2019.

\bibitem{Kam16}
C.\ Kam, S.\ Kompella, G.\ D.\ Nguyen, and A.\ Ephremides, 
``Effect of Message Transmission Path Diversity on Status Age,'' 
\emph{IEEE Trans.\ Inf.\ Theory}, vol.\ 62, no.\ 3, pp.\
1360--1374, 2016.

\bibitem{Yates2017}
R.\ Yates, M.\ Tavan, Y.\ Hu, and D.\ Raychaudhuri,
``Timely Cloud Gaming,''
in \emph{Proc.\ IEEE INFOCOM 2017}, 2017.

\bibitem{Inoue2021}
Y.\ Inoue and T.\ Kimura, ``Age-effective information updating over
intermittently connected MANETs,'' {\it IEEE J.\ Sel.\ Areas Commun.}, 
vol.\ 39, no.\ 5, pp.\ 1293--1308, 2021.

\bibitem{Inoue2023}
Y.\ Inoue, K.\ Maruta, and Y.\ Nakayama, 
``Reliable Wireless Networking in Highly Dynamic Environments: Do
Partial Link Statistics Suffice?,'' 
IEEE Trans.\ Commun., vol.\ 71, no.\ 10, pp.\ 6005--6017, 2023.

\bibitem{He16}
Q.\ He, D.\ Yuan, and A.\ Ephremides, 
``On Optimal Link Scheduling with Min-Max Peak Age of Information in
Wireless Systems,'' 
in \emph{Proc.\ of IEEE ICC 2016}, 2016.

\bibitem{Hsu20}
Y.\ -P.\ Hsu, E.\ Modiano, and L.\ Duan, 
``Scheduling Algorithms for Minimizing Age of Information in Wireless
Broadcast Networks with Random Arrivals,'' 
IEEE Trans.\ Mob.\ Comput., vol.\ 19, no.\ 12, pp.\ 2903--2915, 2020.

\bibitem{Sun17-2}
Y.\ Sun, E.\ Uysal-Biyikoglu, R.\ Yates, C.\ M.\ Koksal, and N.\ B.\ Shroff,
``Update or Wait: How to Keep Your Data Fresh,''
\emph{IEEE Trans.\ Inf.\ Theory}, vol.\ 63, no.\ 11, pp.\ 7492--7508, 2017.

\bibitem{Tripathi17}
V.\ Tripathi and S.\ Moharir, 
``Age of Information in Multi-Source Systems,'' 
in \emph{Proc.\  of IEEE GLOBECOM 2017}, 2017.

\bibitem{Jiang18}
Z.\ Jiang, B.\ Krishnamachari, S.\ Zhou, and Z.\ Niu, 
``Can Decentralized Status Update Achieve Universally Near-Optimal
Age-of-Information in Wireless Multiaccess Channels?,'' 
in \emph{Proc.\ of ITC 30}, pp.\ 144--152, 2018.

\bibitem{He18}
Q.\ He, D.\ Yuan, and A.\ Ephremides, 
``Optimal Link Scheduling for Age Minimization in Wireless Systems,'' 
\emph{IEEE Trans.\ Inf.\ Theory}, vol.\ 64, no.\ 7, 
pp.\ 5381--5394, 2018.

\bibitem{Kadota18}
I.\ Kadota, A.\ Sinha, E.\ Uysal-Biyikoglu, R.\ Singh, and E.\ Modiano, 
``Scheduling Policies for Minimizing Age of Information in Broadcast Wireless Networks,'' 
\emph{IEEE/ACM Trans.\ Netw.}, vol.\ 26, no.\ 6, 
pp.\ 2637--2650, 2018.

\bibitem{Kadota19}
I.\ Kadota, A.\ Sinha, and E.\ Modiano, 
``Scheduling Algorithms for Optimizing Age of Information in Wireless
Networks with Throughput Constraints,'' 
\emph{IEEE/ACM Trans.\ Netw.}, vol.\ 27, no.\ 4,

\bibitem{Li19}
C.\ Li, S.\ Li, and Y.\ T.\ Hou, 
``A General Model for Minimizing Age of Information at Network Edge,'' 
in \emph{Proc.\ of IEEE INFOCOM 2019}, pp.\ 118--126, 2019.

\bibitem{Yin19}
B.\ Yin, S.\ Zhang, Y.\ Cheng, L.\ X.\ Cai, Z.\ Jiang, S.\ Zhou, and Z.\ Niu,
``Only Those Requested Count: Proactive Scheduling Policies for
Minimizing Effective Age-of-Information,'' 
in \emph{Proc.\ of IEEE INFOCOM 2019}, pp.\ 109--117, 2019.

\bibitem{Kaul12-2}
S.~Kaul, R.~Yates, and M.~Gruteser,
``Status Updates through Queues,''
in \emph{Proc.\ IEEE CISS}, 2012.

\bibitem{Campati19}
J.~P.~Champati, H.~Al-Zubaidy, and J.~Gross,
``On the Distribution of AoI for the GI/GI/1/1 and GI/GI/1/2* Systems: Exact Expressions and Bounds,''
in \emph{Proc.\ IEEE INFOCOM}, pp.~37--45, 2019.

\bibitem{Moltafet20}
M.~Moltafet, M.~Leinonen, and M.~Codreanu,
``On the Age of Information in Multi-Source Queueing Models,''
\emph{IEEE Trans.\ Commun.}, vol.~68, no.~8, pp.~5003--5017, 2020.

\bibitem{Inoue2024}
Y.~Inoue and T.~Takine,
``Exact Analysis of the Age of Information in the Multi-Source M/GI/1 Queueing System,''
\emph{arXiv preprint arXiv:2404.05167}, 2024.

\bibitem{Inoue2025}
Y.~Inoue and M.~Mandjes, 
``Characterizing the Age of Information with Multiple Coexisting Data Streams,'' 
To appear in IEEE Trans.\ Inf.\ Theory, 2025.

\bibitem{Chen2022}
Z.~Chen, D.~Deng, C.~She, Y.~Jia, L.~Liang, S.~Fang, M.~Wang, and Y.~Li,
``Age of Information: The Multi-Stream M/G/1/1 Non-Preemptive System,''
\emph{IEEE Trans.\ Commun.}, vol.~70, no.~4, pp.~2328--2341, 2022.

\bibitem{Jiang2021}
Y.~Jiang and N.~Miyoshi,
``Joint Performance Analysis of Ages of Information in a Multi-Source Pushout Server,''
\emph{IEEE Trans.\ Inf.\ Theory}, vol.~68, no.~2, pp.~965--975, 2022.

\bibitem{Liu2021}
Z.~Liu, Y.~Sang, B.~Li, and B.~Ji,
``A Worst-Case Approximate Analysis of Peak Age-of-Information via Robust Queueing Approach,''
in \emph{Proc.\ IEEE INFOCOM}, 2021.

\bibitem{Yates2020}
R.~D.~Yates,
``The Age of Information in Networks: Moments, Distributions, and Sampling,''
\emph{IEEE Trans.\ Inf.\ Theory}, vol.~66, no.~9, pp.~5712--5728, 2020.

\bibitem{Sun2018}
Y.~Sun, E.~Uysal-Biyikoglu, and S.~Kompella,
``Age-optimal Updates of Multiple Information Flows,''
in \emph{Proc.\ 1st Int.\ Workshop on AoI (AoI'18)}, pp.~136--141, 2018.

\bibitem{Bedewy2021}
A.~M.~Bedewy, Y.~Sun, S.~Kompella, and N.~B.~Shroff,
``Optimal Sampling and Scheduling for Timely Status Updates in Multi-Source Networks,''
\emph{IEEE Trans.\ Inf.\ Theory}, vol.~67, no.~6, pp.~4019--4034, 2021.

\bibitem{Bacinoglu15}
B.~T.~Bacinoglu, E.~T.~Ceran, and E.~Uysal-Biyikoglu,
``Age of Information Under Energy Replenishment Constraints,''
in \emph{Proc.\ ITA}, pp.~25--31, 2015.

\bibitem{Yates15}
R.~D.~Yates,
``Lazy is Timely: Status Updates by an Energy Harvesting Source,''
in \emph{Proc.\ IEEE ISIT}, pp.~3008--3012, 2015.

\bibitem{Arafa17}
A.~Arafa and S.~Ulukus,
``Age-Minimal Transmission in Energy Harvesting Two-Hop Networks,''
in \emph{Proc.\ IEEE GLOBECOM}, 2017.

\bibitem{Bacinoglu17}
B.~T.~Bacinoglu and E.~Uysal-Biyikoglu,
``Scheduling Status Updates to Minimize Age of Information with an Energy Harvesting Sensor,''
in \emph{Proc.\ IEEE ISIT}, pp.~1122--1126, 2017.

\bibitem{Wu18}
X.~Wu, J.~Yang, and J.~Wu,
``Optimal Status Update for Age of Information Minimization with an Energy Harvesting Source,''
\emph{IEEE Trans.\ Green Commun.\ Netw.}, vol.~2, no.~1, pp.~193--204, 2018.

\bibitem{Arafa18}
A.~Arafa, J.~Yang, and S.~Ulukus,
``Age-Minimal Online Policies for Energy Harvesting Sensors with Random Battery Recharges,''
in \emph{Proc.\ IEEE ICC}, 2018.

\bibitem{Bacinoglu18}
B.~T.~Bacinoglu, Y.~Sun, E.~Uysal-Biyikoglu, and V.~Mutlu,
``Achieving the Age-Energy Tradeoff with a Finite-Battery Energy Harvesting Source,''
in \emph{Proc.\ IEEE ISIT}, pp.~876--880, 2018.

\bibitem{Farazi18}
S.~Farazi, A.~G.~Klein, and D.~R.~Brown,
``Age of Information in Energy Harvesting Status Update Systems: When to Preempt in Service?,''
in \emph{Proc.\ IEEE ISIT}, pp.~2436--2440, 2018.

\bibitem{Hirosawa2020}
N.~Hirosawa, H.~Iimori, K.~Ishibashi, and G.~T.~F.~de Abreu,
``Minimizing Age of Information in Energy Harvesting Wireless Sensor Networks,''
\emph{IEEE Access}, vol.~8, pp.~219934--219945, 2020.

\bibitem{Li2022}
C.~Li, Q.~Liu, S.~Li, Y.~Chen, Y.~T.~Hou, W.~Lou, and S.~Kompella,
``Scheduling With Age of Information Guarantee,''
\emph{IEEE/ACM Trans.\ Netw.}, vol.~30, no.~5, pp.~2046--2059, 2022.

\bibitem{Maatouk2022}
A.~Maatouk, Y.~Sun, A.~Ephremides, and M.~Assaad,
``Timely Updates With Priorities: Lexicographic Age Optimality,''
\emph{IEEE Trans.\ Commun.}, vol.~70, no.~5, pp.~3020--3033, 2022.

\bibitem{Zakeri2024}
A.~Zakeri, M.~Moltafet, M.~Leinonen, and M.~Codreanu,
``Minimizing the AoI in Resource-Constrained Multi-Source Relaying Systems: Dynamic and Learning-Based Scheduling,''
\emph{IEEE Trans.\ Wireless Commun.}, vol.~23, no.~1, pp.~450--466, 2024.

\bibitem{shaked2006}
M.\ Shaked and J.\ G.\ Shanthikumar, {\it Stochastic Orders,} Springer, New York, NY, 2006.

\end{thebibliography}
\end{document}